\documentclass[11pt]{article}

\usepackage{amsmath,amsthm,amssymb}
\usepackage{fullpage}
\usepackage{hyperref}
\usepackage{authblk}
\usepackage{amssymb} 
\usepackage{verbatim} 
\usepackage{footnote} 
\usepackage{xspace}
\usepackage{tikz}
\usepackage{caption}
\usepackage{subcaption}
\usepackage{cleveref}
\usepackage{algorithm}
\usepackage{algorithmic}
\usepackage{float}
\usepackage{verbatim} 
\usepackage{footnote} 
\usepackage{xspace}
\usepackage{tikz}
\usepackage{xcolor}
\usepackage{colorspace}
\usepackage{caption}
\usepackage{tablefootnote}
\usepackage{enumerate}
\usepackage{mathtools}
\usepackage[most]{tcolorbox}
\usepackage{standalone}
\usepackage[T1]{fontenc}
\usepackage{lmodern}

\newcommand{\prob}[1]{\textup{\textsc{\lowercase{#1}}}\xspace}
\newcommand{\Oh}{\mathcal{O}} 

\newcommand{\N}{\mathbb{N}}

\newcommand{\APSP}{\prob{All-Pairs Shortest Paths}}
\newcommand{\VWAPSP}{\prob{Vertex-Weighted All-Pairs Shortest Paths}}
\newcommand{\TC}{\prob{Triangle Counting}}
\newcommand{\NCD}{\prob{Negative Cycle Detection}}

\DeclareMathOperator{\td}{\mathsf{td}}

\DeclareMathOperator{\mw}{\mathsf{mw}}
\DeclareMathOperator{\mtd}{\mathsf{mtd}}

\DeclareMathOperator*{\argmin}{arg\,min}

\newcommand{\Subst}{\textup{\textsc{Subst}}} 
\newcommand{\SubstMTD}{\textup{\textsc{SubstTd}}}
\newcommand{\Inc}{\textup{\textsc{Inc}}} 
\newcommand{\Union}{\textup{\textsc{Union}}}
\newcommand{\Jn}{\textup{\textsc{Join}}}
\newcommand{\IncxEx}{\Inc_{x,E_x}}
\newcommand{\SubstH}{\Subst_H}
\newcommand{\Hw}{H_w}
\newcommand{\n}[1]{n_#1}
\newcommand{\m}[1]{m_#1}

\newcommand{\T}[1]{T^{#1}} 
\newcommand{\Tv}[2]{T^{#1}_{#2}} 
\newcommand{\Gv}[2]{G^{#1}_{#2}}
\newcommand{\U}{S}
\newcommand{\expr}{\sigma} 
\newcommand{\shiftedw}{c^w}
\newcommand{\class}[1]{\textup{\textsc{#1}}}
\newcommand{\Td}[1]{\class{Td}_{#1}}
\newcommand{\Mw}[1]{\class{Mw}_{#1}}
\newcommand{\Mtd}[1]{\class{Mtd}_{#1}}
\newcommand{\TdMw}[2]{\Td{#1}\Mw{#2}}
\newcommand{\TdMtd}[2]{\Td{#1}\Mtd{#2}}
\newcommand{\MwMtd}[2]{\Mw{#1}\Mtd{#2}}
\newcommand{\TdMwMtd}[3]{\Td{#1}\Mw{#2}\Mtd{#3}}
\newcommand{\Tdk}{\Td{k}}
\newcommand{\Mwh}{\Mw{h}}
\newcommand{\Mtdl}{\Mtd{\ell}}
\newcommand{\TdkMwh}{\TdMw{k}{h}}
\newcommand{\TdkMtdl}{\TdMtd{k}{\ell}}
\newcommand{\MwhMtdl}{\MwMtd{h}{\ell}}
\newcommand{\TdkMwhMtdl}{\TdMwMtd{k}{h}{\ell}}
\newcommand{\disjointunion}{\mathop{\dot{\cup}}}
\newcommand{\joinop}{\Join}

\newtheorem{theorem}{Theorem}
\newtheorem{lemma}{Lemma}

\newtheorem{proposition}[theorem]{Proposition}

\newtheorem{corollary}[theorem]{Corollary}

\newtheorem{remark}[theorem]{Remark}

\newtheorem{observation}[theorem]{Observation}

\theoremstyle{definition}
\newtheorem{mydefinition}[lemma]{Definition}

\date{}
\makeatletter 
\renewcommand\subparagraph{\@startsection{subparagraph}{5}{0pt}%
                                       {3.25ex \@plus1ex \@minus .2ex}%
                                       {-1em}%
                                      {\normalfont\normalsize\bfseries}} 
\makeatother 

\title{Efficient parameterized algorithms on graphs with heterogeneous structure: Combining tree-depth and modular-width}

\author[1]{Stefan Kratsch}
\author[1]{Florian Nelles}

\affil[1]{\small Department of Computer Science, Humboldt-Universit{\"a}t zu
Berlin, Germany 
{\{kratsch,nelles\}@informatik.hu-berlin.de}}

\begin{document}

\maketitle

\begin{abstract}
Many computational problems admit fast algorithms on special inputs, however, the required properties might be quite restrictive. E.g., many graph problems can be solved much faster on interval or cographs, or on graphs of small modular-width or small tree-width, than on general graphs. One challenge is to attain the greatest generality of such results, i.e., being applicable to less restrictive input classes, without losing much in terms of running time.

Building on the use of algebraic expressions we present a clean and robust way of combining such homogeneous structure into more complex heterogeneous structure, and we show-case this for the combination of modular-width, tree-depth, and a natural notion of modular tree-depth. We give a generic framework for designing efficient parameterized algorithms on the created graph classes, aimed at getting competitive running times that match the homogeneous cases. To show the applicability we give efficient parameterized algorithms for \NCD, \VWAPSP, and \TC.
\end{abstract}

\section{Introduction}

Most computational problems can be solved (much) faster on inputs that exhibit certain beneficial structure, such as symmetry, sparsity, structured separations, or graphs of bounded parameter values like bounded tree-depth, clique-width, or modular-width, than they can be solved in general; this is true for both tractable and intractable problems. 
As a downside, the required structure is often not very general as well as many graph parameters are incomparable to one another, e.g., the class all subdivided stars is of bounded tree-depth though of unbounded modular-width, in contrast to the class of all cliques $K_n$ that is of unbounded tree-depth but bounded modular-with; while for both parameters there are efficient parameterized algorithms for many problems.
Thus, for a specific problem, often there just is no \emph{most general} parameter for which this problem can be solved more efficiently, or there is one, but the
dependency on the parameter is much worse than on smaller and incomparable parameters.
Considering the wealth of such algorithmic results, our motivation is to define graph classes generalizing two or more parameters while still being able to use the beneficial structure to obtain efficient algorithms.

In this work, building on measures of the graph structure tree-depth and modular-width, we show how to robustly define \emph{heterogeneous combinations} of such measures and how to (often optimally) use the corresponding graph structure for faster and more general algorithms. To this end, we adopt definitions based on graph operations and algebraic expressions such as are common for clique-width and as were used by Iwata et al.~\cite{IwataOO18} for tree-depth alone. By allowing richer sets of operations, stemming from different types of beneficial structure, we robustly define a large range of heterogeneous graph structure that admits faster algorithms than the general case. Apart from algebraic expressions for clique-width and tree-depth, we also find motivation through modular tree-width, introduced by Paulusma et al.~\cite{PaulusmaSS16}, which is the tree-width of the graph after contracting all twin classes. In a similar (but more general) style we define modular tree-depth to extend modular-width by allowing substitutions into (possibly large) graphs of small tree-depth rather than graphs of small size.

\subparagraph{Our work.}
Our conceptual contribution is a clean and robust way of formalizing classes of graphs with (possibly beneficial) heterogeneous structure by using an operations-based perspective on forms of homogeneous structure along with algebraic expressions. In the present work, these build on operations used for constructing graphs of small tree-depth, modular-width, or modular tree-depth, but the potential for encompassing a much greater variety of operations (and hence greater variety structure) is evident. We show formally, how the arising graph classes relate to one another, establishing for example that the combination of tree-depth and modular-width is incomparable to modular tree-depth. Similarly, already for bounded parameter values, the new forms of heterogeneous structure are incomparable to underlying forms of homogeneous structure.

On the algorithmic side, we greatly extend the framework of Iwata et al.~\cite{IwataOO18}, which applies to tree-depth via the operations union and addition of a vertex, to work for all required operations. Building on the work of Kratsch and Nelles \cite{KratschN20}, we give a general running-time framework that simplifies the task of obtaining running times that match the known bounds for the included homogeneous cases. To show its applicability, we apply our framework to three example problems, namely \NCD, \VWAPSP, and \TC. For each problem, using our framework, we give algorithmic results relative to heterogeneous measures that throughout match the best running time known for the homogeneous case. See Table~\ref{Tab:OverviewRT} for an overview of our results.

\begin{table}[t]
\renewcommand{\arraystretch}{1.3}
\caption{Overview of our algorithmic results, where $n$ and $m$ denote the number of vertices and edges of the input graph. \NCD and \VWAPSP are on directed, vertex-weighted graphs.}
 \begin{tabular}{@{$\;$}r@{$\;$}|@{$\;$}l@{$\;$}|@{$\;$}l@{$\;$}|@{$\;$}l@{$\;$}}
Graph class & \TC          & \NCD &\prob{Vert.-Wght.\ APSP}\\ 
 \hline
 &&&\\[-1.1em]
 $\Mwh$         & $\Oh(h^{\omega-1}n + m)$ \hfill \cite{KratschN18}&  $\Oh(h^2 n + n^2)$                       & $\Oh(h^2 n + n^2)$ \\
 $\Tdk$         & $\Oh(k m)$                                &  $\Oh(k(m + n \log n))$ \hfill \cite{IwataOO18}  
                                                                                                        & $\Oh(k n^2)$ \\
 $\Mtdl$        & $\Oh(\ell m)$                             &  $\Oh(\ell(m+n\log n) + n^2)$             & $ \Oh(\ell n^2)$ \\
 $\TdkMwh$      & $\Oh(k m + h^{\omega - 1}n)$              &  $\Oh(k (m + n \log n)+h^2n + n^2) $      & $\Oh(k n^2 + h^2 n)$ \\
 $\TdkMtdl$     & $\Oh((k + \ell) m)$                       &  $\Oh((k + \ell)(m + n \log n ) + n^2)$   & $\Oh((k+ \ell)n^2)$ \\
 $\MwhMtdl$     & $\Oh(h^{\omega-1}n + \ell m)$             &  $\Oh(h^2 n + \ell(m+n\log n) + n^2 )$    & $\Oh(h^2 n + \ell n^2)$ \\
 $\TdkMwhMtdl$  & $\Oh(h^{\omega-1}n + (k+\ell) m)$         &  $\Oh(h^2 n + (k+\ell)(m+n\log n)+n^2 )$  & $\Oh(h^2 n + (k+\ell) n^2)$ 
 \end{tabular}\label{Tab:OverviewRT}
\end{table}

\subparagraph{Related work.}
Initiated by the work of Giannopoulou et al.~\cite{GiannopoulouMN15}, many publications deal with efficient parameterized algorithms for tractable problems, also called ``FPT in P'', for several parameters and different problems~\cite{FominLSPW18,MertziosNN16,BentertN19,FluschnikKMNNT17,IwataOO18,CoudertDP18,KratschN18,KratschN20,DucoffeP21}. 
Iwata et al.~\cite{IwataOO18} have considered in their work several problems parameterized by the parameter tree-depth, among those a $\Oh(k (m + n\log n))$-time algorithm for \NCD, where $k$ denotes the tree-depth of the input graph. An $\Oh(h^2n + n^2)$-time algorithm to solve \VWAPSP on graphs with modular-width $h$ was presented by Kratsch and Nelles \cite{KratschN20}, but it applies only to non-negative vertex weights.

The modular tree-width of a graph was first considered as a parameter for CNF formulas in several works \cite{PaulusmaSS16,Mengel16,LampisM17} and was defined as the tree-width after the contraction of modules of the incidence graph of the formula. 
On general graphs the modular tree-width was considered by Lampis~\cite{Lampis20}, defined as the tree-width of the graph obtained from a graph if one collapses each twin class into a single vertex.

A different way of combining tree-depth and modular-width lies in the notion of shrub-depth, introduced by Ganian et al.~\cite{GanianHNOM19}, which, in fact, is a tree-depth-like variation of clique-width. Graph classes of bounded tree-depth also have bounded shrub-depth while also dense graph classes, e.g.\ the class of all cliques, have bounded shrub-depth.
Since cographs do not have bounded shrub-depth, all graph classes with heterogeneous structure introduced in this work are incomparable to shrub-depth. Moreover, to our best knowledge, there are not yet any dedicated efficient parameterized algorithms for graphs of small shrub-depth.

Let us also point out that for any fixed values of $k$, $\ell$, and $h$ all classes introduced in this work have bounded clique-width. That said, the clique-width may be much larger than these parameters and the best running times relative to clique-width can be worse (and they are for the examples provided in this work, cf.~\cite{CoudertDP18,Ducoffe20}). In other words, we obtain faster algorithms by making use of the specific (heterogeneous) structure of the introduced graph classes, rather than considering them as graphs of bounded clique-width. (Similarly, all considered problems of course also have efficient algorithms depending on just $n$ and $m$.)

\subparagraph{Organization.}
In Section~\ref{section:preliminaries} we recall operations-based definitions for tree-depth and modular-width. In Section~\ref{section:heterogeneousstructure} we introduce modular tree-depth and use algebraic expressions to define graph classes with heterogeneous structure. The running-time framework is presented in Section~\ref{section:framework} and the applications to \TC, \NCD, and \VWAPSP can be found in Section~\ref{section:applications}. The relations between the obtained graph classes are explored in Section~\ref{section:relations}. We conclude in Section~\ref{section:conclusion}.

\section{Preliminaries}\label{section:preliminaries}

\subparagraph{Notation.}
We consider only simple graphs, i.e., graphs with no loops and no parallel edges. For a graph $G$, we denote by $V(G)$ the set of vertices and by $E(G)$ the set of edges. We shorthand $|G|$ for $\vert V(G)\vert$. 
For a graph $G$ and edge weights $c \colon E(G) \rightarrow \mathbb{R}$ 
we define $dist_{G,c}(u,v)$ as the shortest path distance from the vertex $u$ to $v$ w.r.t.~the edge weights $c$ for $u,v \in V(G)$. For vertex weights $w \colon V(G) \rightarrow \mathbb{R}$, we define $dist_{G,w}(u,v)$  analogously and we may omit $G$ resp.\ $w$ in the subscript if the graph resp.\ the weight function $w$ is clear from the context. 

For two graphs $G_1 = (V_1, E_1)$ and $G_2 = (V_2, E_2)$, we denote by $G_1\disjointunion G_2$  their disjoint union, i.e., the graph $(V_1\mathop{\dot{\cup}}V_2, E_1 \mathop{\dot{\cup}} E_2)$ and we denote by $G_1\joinop G_2$ their (disjoint) join, i.e., the graph $(V_1 \mathop{\dot{\cup}} V_2, E_1 \mathop{\dot{\cup}} E_2 \mathop{\dot{\cup}} \{\{u,v\} \mid u \in V_1, v \in V_2\}$. 
For two graphs $G$ and $H$, and a vertex $v$ of $H$, by $H[v\leftarrow G]$ we denote the substitution of $G$ into $v$ in $H$, i.e., the graph obtained by replacing $v$ in $H$ with the graph $G$ and giving each of its vertices the same neighborhood as $v$ had. This extends up to complete substitution into a $t$-vertex graph $H$, denoted $H[v_1\leftarrow G_1,\ldots,v_t\leftarrow G_t]$. We shorthand $[k]$ for $\{1,\ldots,k\}$.

\subparagraph{Graph operations.}
These graph operations will be used in algebraic expressions. The nullary operations $\circ$ and $\bullet$ return the empty graph respectively the graph with a single vertex. For $t\geq 2$, the $t$-ary operations $\Union_t$ and $\Jn_t$ are defined by $\Union_t(G_1,\ldots,G_t):=G_1\disjointunion\ldots\disjointunion G_t$ and $\Jn_t(G_1,\ldots,G_t):=G_1\joinop\ldots\joinop G_t$; we will usually omit the subscript $t$ and allow these operations for all $t$. For a vertex $x$ and a set $E_x\subseteq\{\{x,v\} \mid v \in V\}$, the unary operation $\IncxEx$ is defined by $\IncxEx(G)=(V\cup\{x\},E\cup E_x)$ for all graphs $G=(V,E)$ with $x\notin V$.
For a graph $H$ with $V(H)=\{v_1,\ldots,v_t\}$, the $t$-ary operation $\SubstH$ is defined by $\SubstH(G_1,\ldots,G_t):=H[v_1\leftarrow G_1,\ldots,v_t\leftarrow G_t]$. 

\subparagraph{Algebraic expressions.}
As is common especially for clique-width, we use algebraic expressions over certain sets of operations to describe the structure or the construction of graphs. We say that $G$ has an algebraic expression if $G$ is the result of evaluating the expression (possibly followed by a single renaming of vertices). 
For an algebraic expression $\expr$, we denote by $val(\expr)$ the resulting graph and by $\T{\expr}$ we denote the corresponding expression tree, i.e., a rooted tree in which each node corresponds to an operation of the expression that is applied to its children.
We define the \emph{nesting depth} of an operation in an expression tree as the maximum number of nodes in the expression tree corresponding to this operation that are on a root-to-leaf path. We denote the nesting depth of an operation in an algebraic expression as the nesting depth of this operation in the corresponding expression tree.
We let the empty graph correspond to the empty expression, irrespective of the set of operations. 

\subparagraph{Tree-depth.}
There are many equivalent definitions for the tree-depth of a graph, e.g., as the minimum height of a rooted forest whose closure contains $G$ as a subgraph or by a recursive definition based on connected components and vertex deletion \cite{NesetrilM12, GanianHKLOR09}.
We give a folklore definition via algebraic expressions.

\begin{mydefinition}\label{definition:treedepth} 
 The \emph{tree-depth} of a graph $G$, denoted $\td(G)$, is the smallest $k\in\N$ such that $G$ has an algebraic expression over $\{\circ,\Union\}\cup\{\IncxEx\}$ whose $\IncxEx$ operations have nesting depth at most $k$. The class $\Tdk$ contains all graphs of tree-depth at most $k$.
\end{mydefinition}

\subparagraph{Modular-width.}
Modular-width was first mentioned by Courcelle and Olariu~\cite{CourcelleO00}. 
A \emph{module} of a graph $G$ is a set $M \subseteq V(G)$, such that each vertex in $V(G) \setminus M$ is either connected to none or to all vertices in $M$. A \emph{modular partition} of a graph is a partition of the vertices into modules. The modular-width roughly indicate how well a graph can be recursively partitioned into modules.
Like tree-depth it has several concrete definitions, which unfortunately do not all agree on what graphs have modular-width zero, one, or two; they are equivalent for modular-width $h$, for all $h\geq 3$, so there is no asymptotic difference between them.\footnote{A frequent variant is to say that the modular-width of any graph is the least integer $h\geq 2$ such that any prime subgraph in its modular decomposition has order at most $h$ (cf.~\cite{CoudertDP18}), while Abu-Khzam et al.~\cite{Abu-KhzamLMHP17} define it to be the maximum degree of the modular decomposition tree.} 

We mildly adapt a definition due to Gajarsk\'y et al.~\cite{GajarskyLO13} to our needs. For example, note that $\Mw{0}$ is the class of cographs.

\begin{mydefinition}[adapted from \cite{GajarskyLO13}]\label{definition:modularwidth}
 The \emph{modular-width} of a graph $G$, denoted $\mw(G)$, is the smallest $h\in\N$ such that $G$ has an algebraic expression over $\{\bullet,\Union,\Jn\}\cup\{\SubstH\mid |H|\leq h\}$. The class $\Mwh$ contains all graphs of modular-width at most $h$.
\end{mydefinition}

\section{Heterogeneous structure}\label{section:heterogeneousstructure}

Here we use the introduced graph operations to define classes of graphs with heterogeneous structure that generalize both $\Tdk$ and $\Mwh$. We will also use them to encompass and generalize modular tree-depth, which we will introduce in a moment. In the following section, we then show how to use such structure algorithmically (and optimally).

\begin{mydefinition}\label{definition:tdkmwh}
For $k,h\in\N$, the class $\TdkMwh$ contains all graphs $G$ that have an algebraic expression over $\{\circ,\bullet,\Union,\Jn\}\cup\{\IncxEx\}\cup\{\SubstH\mid|H|\leq h\}$ whose $\IncxEx$ operations have nesting depth at most $k$.
\end{mydefinition}

The following propositions follow directly from the definition of $\TdkMwh$; similar relations are true for the other classes and we do not list all of them explicitly. The inequality in Proposition~\ref{proposition:tdkmwh:smallvalues} holds since all cliques are contained in the class $\Mw{0}$ and therefore in $\TdMw{k}{0}$. The converse non-relations, even for bounded values of $k$ and $h$, are showed later.

\begin{proposition}\label{proposition:tdkmwh:relations}
$\Tdk\subseteq\TdkMwh$ and $\Mwh\subseteq\TdkMwh$. 
\end{proposition}

\begin{proposition}\label{proposition:tdkmwh:smallvalues}
$\TdMw{0}{h}=\Mwh$ but $\TdMw{k}{0}\neq\Tdk$
\end{proposition}

\subparagraph{Modular tree-depth.}
Modular tree-width was introduced and studied in several recent papers \cite{PaulusmaSS16,Mengel16,LampisM17,Lampis20}. In these works, the modular tree-width of a graph $G$ is the smallest $\ell\in\N$ such that $G$ can be constructed from a graph of tree-width at most $\ell$ by substituting each vertex with an independent set or a clique of arbitrary size. In other words, the modular tree-width is the width of the graph after collapsing each \emph{twin class} to a single vertex. (Here two vertices are \emph{twins} if they have the same sets of neighbors; clearly, this is an equivalence relation. The twin classes are the equivalence classes of the twin relation.)

An analogous definition for modular tree-depth would entail substituting cliques and independent sets into a graph of tree-depth at most $\ell$. In our definition of modular tree-depth we deviate from this style, by instead extending modular-width to allow substitution into pattern graphs $H$ of arbitrary size but tree-depth at most~$\ell$. (We would similarly define (generalized) modular tree-width but we do not study it in this work.) To avoid confusion, we will use \emph{restricted modular tree-depth} to refer to modular tree-depth defined in the above style. Note that in many of the mentioned applications, restricted and generalized modular tree-width coincide because the considered graphs have a modular partition whose modules are cliques and independent sets.

\begin{mydefinition}[modular tree-depth]\label{definition:modulartreedepth}
 The \emph{modular tree-depth} of a graph $G$, denoted $\mtd(G)$, is the smallest $\ell\in\N$ such that $G$ has an algebraic expression over $\{\bullet,\Union,\Jn\}\cup\{\SubstH\mid \td(H)\leq \ell\}$. The class $\Mtdl$ contains all graphs of modular tree-depth at most $\ell$.
\end{mydefinition}

\subparagraph{Further classes with heterogeneous structure.}
Now we can define three natural combinations of $\Mtdl$ with the classes considered so far. Note that $\Mtdl$ subsumes $\Tdk$ and $\Mwh$ for $\ell\geq k,h$ but, surprisingly perhaps, there is no $\ell\in\N$ such that it fully contains $\TdMw{1}{0}$. This also means that, e.g., $\MwhMtdl\subseteq\Mtdl$ when $h\leq \ell$ but in general this is not the case. Intuitively, substitution into a pattern $H$ of small tree-depth is algorithmically more costly than substitution into a small pattern $H$, hence the case $h>\ell$ is sensible. The relations between these graph classes are explored in Section~\ref{section:relations}.

\begin{mydefinition}\label{definition:furtherclasses}
 For $k,h,\ell\in\N$ we define the three classes $\TdkMtdl$, $\MwhMtdl$, and $\TdkMwhMtdl$ to contain all graphs that have an algebraic expression of the following type:
 \begin{itemize}
  \item $\TdkMtdl$: algebraic expressions over $\{\circ,\bullet,\Union,\Jn\}\cup\{\IncxEx\}\cup\{\SubstH\mid \td(H)\leq \ell\}$ whose $\IncxEx$ operations have nesting depth at most $k$
  \item $\MwhMtdl$: algebraic expressions over $\{\bullet,\Union,\Jn\}\cup\{\SubstH\mid |H|\leq h\}\cup\{\SubstH\mid \td(H)\leq \ell\}$
  \item $\TdkMwhMtdl$: algebraic expressions over $\{\circ,\bullet,\Union,\Jn\}\cup\{\IncxEx\}\cup\{\SubstH\mid |H|\leq h\}\cup\{\SubstH\mid \td(H)\leq \ell\}$ whose $\IncxEx$ operations have nesting depth at most~$k$
 \end{itemize}
\end{mydefinition}

\begin{remark}
The sets of allowed operations are simply the unions of what is allowed in the homogeneous case, while $\IncxEx$ operations keep their restriction on nesting depth.
\end{remark}

\begin{remark}
In this work, we always implicitly assume that for the operations $\{ \SubstH \mid \td(H) \leq \ell\}$ a tree-depth expression for $H$ is given. 
\end{remark}

This concludes the introduction of new graph classes with heterogeneous structure. We will now turn to showing how useful they are for designing more general algorithms for well-structured graphs. Fortunately, this turns out to be just as robust as the simpler case of the class $\Tdk$, studied by Iwata et al.~\cite{IwataOO18}, and comes down to designing a separate routine for each allowed operation.

\section{Running Time Framework}\label{section:framework}

In this section, we combine the running time framework provided by Kratsch and Nelles~\cite{KratschN20} for graphs parameterized by the modular-width with the divide-and-conquer framework proved by Iwata et al.~\cite{IwataOO18} for graphs parameterized by the tree-depth.
We will restrict ourselves to running times $T \colon \mathbb{R}_{\geq 1} \rightarrow \mathbb{R}_{\geq 1}$ that are \emph{superhomogeneous} \cite{BuraiSA05}, i.e., for all $\lambda \geq 1$ it holds that $\lambda \cdot T(n) \leq T(\lambda \cdot n)$. 
The following lemma was obtained in \cite{KratschN20}.

\begin{lemma}[\cite{KratschN20}]\label{lem:properySuperhomogeneous}
 Let $T \colon \mathbb{R}_{\geq 1}^2 \rightarrow \mathbb{R}_{\geq 1}$ be a function that is superhomogeneous in the first component and  monotonically increasing in the second component. Then
 \begin{align*}
 \max_{ \substack{1 \leq x \leq n \\ 1 \leq y \leq m}} \frac{T(x,y)}{x} \leq \frac{T(n,m)}{n}.
 \end{align*}
\end{lemma}

Inspired by the functional way of the divide-and-conquer framework by Iwata et al.~\cite{IwataOO18}, we extend this approach to cope with the operations defined in the previous sections. 

\begin{theorem}\label{thm:runningtime}
Let $G$ be a graph such that $G \in \TdkMwhMtdl$ with a given expression for some integers $k,h,\ell \in \mathbb{N}$ and let $f$ be a function defined on subgraphs of $G$ such that $f(\circ)$ and $f(\bullet)$ can be computed in constant time. Let further $T_{\mathrm{Inc}}$, $T_{\mathrm{Sub}}$, and $T_{\mathrm{SubTd}}$ be functions that are superhomogeneous in each component that bound the running times of the following algorithms $A_{\mathrm{Inc}}$, $A_{\mathrm{Sub}}$, and $A_{\mathrm{SubTd}}$:
\begin{itemize}
\item $A_{\mathrm{Inc}}(G',f(G'),x,E_x) \mapsto f(\IncxEx(G'))$. Given a graph $G'$, its value $f(G')$, a vertex $x \notin V(G')$, and $E_x \subseteq \{\{x,u\} \mid u \in V(G')\}$, this algorithm computes the value $f(\IncxEx(G'))$ in time $T_{\mathrm{Inc}}(\vert V(G') \cup \{x\} \vert,\vert E(G') \cup E_x \vert)$. 

\item $A_{\mathrm{Sub}}(H, (G_1,f(G_1)), \ldots, (G_t,f(G_t))) \mapsto f(\SubstH(G_1, \ldots, G_t))$. Given a $t$-vertex pattern graph $H$ and graphs $G_i$ with their values $f(G_i)$ for $i \in [t]$, this algorithm computes the value $f(\SubstH(G_1, \ldots, G_t))$ in time $T_{\mathrm{Sub}}(\vert V(H) \vert, \vert E(H) \vert) = T_{\mathrm{Sub}}(t,\vert E(H) \vert)$.

\item $A_{\mathrm{SubTd}}(H, (G_1,f(G_1)), \ldots, (G_t,f(G_t))) \mapsto f(\SubstH(G_1, \ldots, G_t))$. Given a $t$-vertex pattern graph $H$ and graphs $G_i$ with their values $f(G_i)$, this algorithm computes the value $f(\SubstH(G_1, \ldots, G_t))$ in time $T_{\mathrm{SubTd}}(\vert V(H) \vert, \vert E(H) \vert,\td (H))$.
\end{itemize}
\noindent
Then, one can compute $f(G)$ in total time $\Oh(k T_{\mathrm{Inc}}(n,m) + \frac{n}{h}  T_{\mathrm{Sub}}(h,m) + T_{\mathrm{SubTd}}(n,m, \ell))$.
\end{theorem}

\begin{proof}
Let $\expr$ be a corresponding algebraic expression of $G\in \TdkMwhMtdl$, i.e., $val(\expr) = G$.
Let $\T{\expr} =tree(\expr)$ be the corresponding expression tree. Clearly, one can replace each occurrence of an operation $\Union_t$ or $\Jn_t$ in $\T{\expr}$ for some $t \geq 1$ by a sequence of $t-1$ operations $\SubstH$ with a pattern graph $H$ consisting of two adjacent or non-adjacent vertices, i.e., $H\in\{I_2,K_2\}$.
Thus, w.l.o.g.~we can assume that $\expr$ does not consists of operations $\Union_t$ nor $\Jn_t$. Furthermore, we can assume that each pattern $H$ for a $\SubstH$ operation is of size at least two and that no argument of any $\SubstH$ operation is the empty graph.
For a node $v \in V(\T{\expr})$, let $\Tv{\expr}{v}$ denote the subtree of $\T{\expr}$ with root node $v$. With a slight abuse of notation, we denote by $val(\Tv{\expr}{v})$ the subgraph of $G$ corresponding to the subexpression tree $\Tv{\expr}{v}$. 

To compute $f(G)$,
we traverse the expression tree $\T{\expr}$ in a bottom-up manner and compute for each subexpression tree $\Tv{\expr}{v}$ of $\T{\expr}$ the value $f(val(\Tv{\expr}{v}))$ for each $v \in V(\T{\expr})$.
See Algorithm~\ref{alg:Algo} for the divide-and-conquer algorithm that computes $f(G)$ by traversing $\T{\expr}$.  
By induction over the length of the expression it is easy to see that Algorithm~\ref{alg:Algo} computes the function $f(G)$ correctly.

\begin{algorithm}[t] 
\caption{Algorithm for computing $f(G)$}
\hspace*{\algorithmicindent} \textbf{Input:} Graph $G$ with corresponding expression tree $\T{\expr}$ \\
 \hspace*{\algorithmicindent} \textbf{Output:} $f(G)$ 
\begin{algorithmic} 
\IF {$G$ is the empty graph}
\STATE \textbf{return} $f(\circ)$
\ENDIF
\IF {$G$ is a single-vertex graph}
\STATE \textbf{return} $f(\bullet)$
\ENDIF
\STATE Let $r$ be the root node of $\T{\expr}$ and let $v_1, \ldots, v_c$ be the children of $r$ in $\T{\expr}$ 
\STATE Let $G_1, \ldots, G_c$ be the graphs $val(\Tv{\expr}{v_1}), \ldots, val(\Tv{\expr}{v_c})$.
\FOR{$i \in [c]$}
\STATE $f(G_i) = Compute(G_i,\Tv{\expr}{v_i})$
\ENDFOR
\STATE Compute $f(G)$ via the algorithm corresponding to the root node $r$.
\STATE \textbf{return} $f(G)$
\end{algorithmic}\label{alg:Algo}
\end{algorithm}

We are left to show the desired running time for Algorithm~\ref{alg:Algo}.
Note that all leaves of $\T{\expr}$ correspond to either $\circ$ or $\bullet$ operations, while in our setting an operation $\circ$ is necessarily followed by an operation $\IncxEx$. Thus, the number of leaves in $\T{\expr}$ is at most $n$, i.e., at most one for each vertex in $V(G)$; the remaining vertices come via additional $\IncxEx$ operations. Since $f(\circ)$ and $f(\bullet)$ can be computed in constant time, the total time for processing all leaves is bounded by $\Oh(n)$.

All interior nodes of $\T{\expr}$ correspond to either to $\IncxEx$ or $\SubstH$ operations. We denote by $V_{\T{\expr}}^{\Inc}$, $V_{\T{\expr}}^{\Subst}$, resp.\ $V_{\T{\expr}}^{\SubstMTD}$ those nodes in $V(\T{\expr})$ that correspond to an operation $\IncxEx$, $\SubstH$ with $\vert H \vert \leq h$, resp.\ $\SubstH$ with $\td(H) \leq \ell$. The total running time of the algorithm can now be bounded by the sum of the running times needed to process each node in $\T{\expr}$. 

For any node $v_H \in V_{\T{\expr}}^{\Subst} \cup V_{\T{\expr}}^{\SubstMTD}$, let $\n{H}$ and $\m{H}$ denote the number of vertices resp.\ edges in the pattern graph $H$ associated with $v_H$. Thus, a node $v_H \in V_{\T{\expr}}^{\Subst} \cup V_{\T{\expr}}^{\SubstMTD}$ has exactly $\n{H}$ children. Note that $\n{H} \geq 2$ for each pattern graph of a node $v_H$ and since the number of leaves in $\T{\expr}$ is bounded by $n$, it holds that $\vert V_{\T{\expr}}^{\Subst} \cup V_{\T{\expr}}^{\SubstMTD}\vert \leq n-1$. Thus, the sum of the values $\n{H}$ for all nodes $v_H \in  V_{\T{\expr}}^{\Subst} \cup V_{\T{\expr}}^{\SubstMTD}$ can be bounded by $2n$, i.e., the number of leaves plus $\vert V_{\T{\expr}}^{\Subst}\cup V_{\T{\expr}}^{\SubstMTD} \vert$.

Thus, we can now bound the combined running time of all nodes in $V_{\T{\expr}}^{\Subst}$ corresponding to an operation $\SubstH$ with $\vert H \vert \leq h$ in a similar way as done in \cite{KratschN20}:
 \begin{align*}
 \sum_{v_H \in V_{\T{\expr}}^{\Subst}} T_{\mathrm{Sub}}(\n{H},\m{H})  &=     \smashoperator[r]{\sum_{v_H \in V_{\T{\expr}}^{\Subst}}} \n{H} \frac{T_{\mathrm{Sub}}(\n{H},\m{H})}{\n{H}}  \\
 		            &\leq \smashoperator[r]{\sum_{v_H \in V_{\T{\expr}}^{\Subst}}} \n{H} \cdot \left( \max_{ \substack{1 \leq \n{H} \leq h \\ 1 \leq \m{H} \leq m}} \frac{T_{\mathrm{Sub}}(\n{H},\m{H})}{\n{H}} \right) \\
 		            &\leq 2n \cdot \frac{T_{\mathrm{Sub}}(h,m)}{h} 
\end{align*}

The final inequality holds due to Lemma~\ref{lem:properySuperhomogeneous} and  due to $T_{\mathrm{Sub}}$ being superhomogeneous.
For all nodes in $V_{\T{\expr}}^{\SubstMTD}$ corresponding to an operation $\SubstH$ with $\td(H) \leq \ell$ the running time can be bounded in a similar way:

 \begin{align*}
 \sum_{v_H \in V_{\T{\expr}}^{\SubstMTD}} T_{\mathrm{SubTd}}(\n{H}, \m{H}, \td(H))  &\leq     \smashoperator[r]{\sum_{v_H \in V_{\T{\expr}}^{\SubstMTD}}} \n{H} \frac{T_{\mathrm{SubTd}}(\n{H},\m{H}, \ell)}{\n{H}}  \\
 		            &\leq \smashoperator[r]{\sum_{v_H \in V_{\T{\expr}}^{\SubstMTD}}} \n{H} \cdot \left( \max_{ \substack{1 \leq \n{H} \leq n \\ 1 \leq \m{H} \leq m}} \frac{T_{\mathrm{SubTd}}(\n{H},\m{H},\ell)}{\n{H}} \right)  \\
 		            &\leq 2n \cdot \frac{T_{\mathrm{SubTd}}(n,m,\ell)}{n} = 2 \cdot T_{\mathrm{SubTd}}(n,m,\ell)
\end{align*}
Finally, consider all nodes in $V_{\T{\expr}}^{\Inc}$.
We define for each node $v \in V_{\T{\expr}}^{\Inc}$ a depth $d(v)$ equal to the maximum number of nodes in $V_{\T{\expr}}^{\Inc}$ on a path from $v$ to a leaf in $\Tv{\expr}{v}$ (including $v$), i.e., the nesting depth of the operation $\IncxEx$ in $\Tv{\expr}{v}$. 
It is easy to see that for two nodes $v, u \in V_{\T{\expr}}^{\Inc}$ with $d(v) = d(u)$ and $v \neq u$, the subtrees $\Tv{\expr}{v}$ and $\Tv{\expr}{u}$ are disjoint. Furthermore, the depth $d(v)$ of any node in $V_{\T{\expr}}^{\Inc}$ is per definition at most $k$. 
For a node $v \in V_{\T{\expr}}^{\Inc}$, let $n_v$ resp.\ $m_v$ denote the number of vertices resp.\ number of edges in $val(\Tv{\expr}{v})$.
We can now upper bound the running time of all nodes $v \in V_{\T{\expr}}^{\Inc}$:
 \begin{align*}
 \sum_{v \in V_{\T{\expr}}^{\Inc}} T_{\mathrm{Inc}}(n_v,m_v)  =  \sum_{i = 1}^k \sum_{\substack{v \in V_{\T{\expr}}^{\Inc} \\ d(v) = i}} T_{\mathrm{Inc}}(n_v,m_v)  
 		            \leq  \sum_{i = 1}^k  T_{\mathrm{Inc}}(n,m)  
 		            \leq k \cdot T_{\mathrm{Inc}}(n,m) \notag
\end{align*}
The penultimate inequality holds due to $T_{\mathrm{Inc}}$ being superhomogeneous (which implies that $T_{\mathrm{Inc}} \in \Omega(n+m)$).
In total, the running time sums up to $\Oh(k T_{\mathrm{Inc}}(n,m) + \frac{n}{h}  T_{\mathrm{Sub}}(h,m) + T_{\mathrm{SubTd}}(n,m,	\ell))$ and we have proven Theorem~\ref{thm:runningtime}.
\end{proof}

Theorem~\ref{thm:runningtime} describes the running time for graphs in the class $\TdkMwhMtdl$ if algorithms $A_{\mathrm{Inc}}$, $A_{\mathrm{Sub}}$, and $A_{\mathrm{SubTd}}$ are known. Similarly, one can state the running times for the other subclasses of $\TdkMwhMtdl$ that are defined in Section~\ref{section:heterogeneousstructure}.

\begin{corollary}\label{cor:runningtimes}
Let $G = (V,E)$ be a graph with $\vert V \vert = n$ and $\vert E \vert = m$. Let further $f$ be a function defined on subgraphs of $G$ and let $T_{\mathrm{Inc}}$, $T_{\mathrm{Sub}}$, resp.\ $T_{\mathrm{SubTd}}$ be as defined in Theorem~\ref{thm:runningtime}. Then, depending on a given expression $\expr$ with $val(\expr) = G$, one can compute $f(G)$ in the following time:
\begin{table}[ht]
\centering
 \begin{tabular}{l|l}
  & Total running time to compute $f(G)$\\
  \hline
  &\\[-1.1em]
  $G \in \Mwh$     & $\Oh(\frac{n}{h}  T_{\mathrm{Sub}}(h,m))$ \hfill \cite{KratschN18}\\
  $G \in\Tdk$     & $\Oh(k T_{\mathrm{Inc}}(n,m))$ \hfill \cite{IwataOO18}\\
  $G \in\Mtdl$    & $\Oh(T_{\mathrm{SubTd}}(n,m,\ell))$\\
  $G \in\TdkMwh$  & $\Oh(k T_{\mathrm{Inc}}(n,m)) + \Oh(\frac{n}{h}  T_{\mathrm{Sub}}(h,m))$ \\
  $G \in\TdkMtdl$ & $\Oh(k T_{\mathrm{Inc}}(n,m)) + \Oh(T_{\mathrm{SubTd}}(n,m,\ell))$            \\
  $G \in\MwhMtdl$ & $\Oh(\frac{n}{h}  T_{\mathrm{Sub}}(h,m)) +  \Oh(T_{\mathrm{SubTd}}(n,m,\ell))$ 
 \end{tabular}
\end{table}
\end{corollary}

\section{Applications}\label{section:applications}

In this section we give applications of Theorem~\ref{thm:runningtime}. We present algorithms solving the problems $\TC$, $\NCD$, and $\APSP$. For the latter two problems we restrict ourselves to vertex-weighted graphs, since for edge-weighted graphs the problems $\NCD$ and $\APSP$ are as hard as the general case already on cliques (cf.~\cite{KratschN18}), which are contained in $\TdMwMtd{0}{0}{0}$.

\subsection{Triangle Counting}\label{sec:TC}
As a first application of Theorem~\ref{thm:runningtime} we consider the problem \prob{triangle counting}.
In the \prob{Triangle Counting} problem we are given an undirected graph $G = (V,E)$ and need to compute the number of triangles in $G$, i.e. the number of $K_3$ subgraphs in $G$.
We will prove the following theorem:

\begin{theorem}\label{thm:TriangleCounting}
Let $G$ be an undirected graph such that $G \in \TdkMwhMtdl$ with a given expression for some $k,h,\ell \in \N$. Then one can solve \prob{Triangle Counting} in time $\Oh(h^{\omega-1}n + (k+\ell) m)$.
\end{theorem}

For a graph $G$, we define $f(G)$ as the function that returns the values $n_G = \vert V(G) \vert$, $m_G = \vert E(G) \vert$, and the number of triangles, denoted by $t_G$.\footnote{For the algorithm $A_{\mathrm{Inc}}$ the value $t_G$ alone would suffice. The values $n_G$ and $m_G$ are only computed to potentially use the algorithms $A_{\mathrm{Sub}}$ and $A_{\mathrm{SubTd}}$.} 
To use Theorem~\ref{thm:runningtime}, we will describe the algorithms $A_{\mathrm{Inc}}$, $A_{\mathrm{Sub}}$, and $A_{\mathrm{SubTD}}$ in the following lemmata.
The following lemma was established by Kratsch and Nelles~\cite{KratschN18}.

\begin{lemma}[\cite{KratschN18}]\label{lem:TCSubst}
Let $H$ be a $t$-vertex graph, let $G_1, \ldots, G_t$ be graphs, and let $f(G_i)$ be given for each $G_i$. Then one can compute $f(\SubstH(G_1, \ldots, G_t))$ in time $\Oh(t^\omega)$.
\end{lemma}

\begin{lemma}\label{lem:TCInc}
Let $G$ be a graph, let $x \notin V(G)$, let $E_x \subseteq \{\{x,u\} \mid u \in V(G)\}$, and let $f(G)$ be given. Then one can compute $f( \IncxEx(G))$ in time $\Oh(\vert E(G) \cup E_x \vert)$.
\end{lemma}

\begin{proof}
Let $n_{G}$, $m_{G}$, and $t_{G}$ be the values returned by $f(G)$.
First, we compute the number of vertices and edges in linear time.
In order to update the number of triangles, we initialize $t = t_{G}$, iterate over all edges in $\IncxEx(G)$, and check if both endpoints are adjacent to $x$. If so, we increment $t$ by one.
Since every triangle that is in $\IncxEx(G)$ but not in $G$ needs to use the vertex $x$ together with exactly one edge in $E(G)$, we have counted every triangle in time $\Oh(\vert E(G) \cup E_x \vert)$.
\end{proof}

\begin{lemma}\label{lem:TCSubstMTD}
Let $H$ be a $t$-vertex graph, let $G_1, \ldots, G_t$ be graphs, and let $f(G_i)$ be given for each $G_i$. Then one can compute $f(\SubstH(G_1, \ldots, G_t))$ in time $\Oh(\td(H) \cdot \vert E(H) \vert)$.
\end{lemma}

\begin{proof}
Let $G = \SubstH(G_1, \ldots, G_t)$ for a $t$-vertex graph $H$ with $V(H) = \{v_1, \ldots, v_t\}$ and let $n_{i}$, $m_{i}$, and $t_{i}$ be the values returned by $f(G_i)$.
As shown in \cite{KratschN18}, for the number of vertices in $G$ it holds that $n_G = \sum_{v_i \in V(H)}n_i$, for the number of edges in $G$ it holds that $m_G = \sum_{v_i \in V(H)} m_i + \sum_{\{v_{i_1}, v_{i_2}\} \in E(H)} n_{i_1} \cdot n_{i_2}$, and finally for the number of triangles it holds that $t_G = \sum_{v_i \in V(H)} t_i + \sum_{\{v_{i_1}, v_{i_2}\} \in E(H)} (m_{i_1}n_{i_2}+n_{i_1}m_{i_2}) +  \sum_{\{v_{i_1}, v_{i_2}, v_{i_3}\} \in T(H)} n_{i_1} \cdot n_{i_2} \cdot n_{i_3}$, where $T(H)$ denotes the set of all triangles in $H$. Thus, $n_G$, $m_G$, and the first two summands of $t_G$ can be computed in time $\Oh(\vert E(H) \vert)$. To compute the final summand of $t_G$, 
we use the tree-depth expression of $H$ by adjusting the algorithm for $\IncxEx$ (shown in Lemma~\ref{lem:TCInc}), such that whenever we find a triangle $\{v_{i_1}, v_{i_2}, v_{i_3}\}$ we increment the number of triangles by $n_{i_1} \cdot n_{i_2} \cdot n_{i_3}$ (instead of by one). Using the tree-depth running-time framework by Iwata et al.~\cite{IwataOO18}, we have proven the lemma.
\end{proof}

\begin{proof}[Proof of Theorem~\ref{thm:TriangleCounting}]
We use Theorem~\ref{thm:runningtime} to prove the claim.
Lemma~\ref{lem:TCInc}, \ref{lem:TCSubst}, resp.\ Lemma~\ref{lem:TCSubstMTD} provide the algorithms $A_{\mathrm{Inc}}$, $A_{\mathrm{Sub}}$, resp.\ $A_{\mathrm{SubTd}}$. Moreover, $f(\circ)$ and $f(\bullet)$ can be computed in constant time. Thus, by Theorem~\ref{thm:runningtime}, the claim follows.
\end{proof}

\subsection{Negative Cycle Detection}\label{sec:NegCycle}

Before stating the algorithm for $\NCD$, we transfer the graph classes defined in Section~\ref{section:heterogeneousstructure} to directed graphs. To do so, we require each pattern graph $H$ to be directed and we define the operation $\IncxEx$ for an edge set $E_x \subseteq \{ (x,u) \mid u \in V\} \cup \{ (u,x) \mid u \in V\}$. 
In this section, we prove the following theorem.

\begin{theorem}\label{thm:NegCycleDetection}
Let $G = (V,E)$ be a directed graph such that $G \in \TdkMwhMtdl$ with a given expression for some $k,h,\ell \in \N$, and let $w \colon V \rightarrow \mathbb{R}$. Then one can solve \NCD combinatorially in time $\Oh(h^2 n + (k+\ell)(m+n\log n) + n^2)$ resp. in time $\Oh(h^{1.842} n + (k+\ell)(m+n\log n) + n^2)$ using fast matrix multiplication.
\end{theorem}

First, we need to recall some notations about (feasible) potentials.
For a directed graph $G= (V,E)$ with edge weights $c \colon E \rightarrow \mathbb{R}$ and a function $\pi \colon V \rightarrow \mathbb{R}$, we define $c_\pi \colon E \rightarrow \mathbb{R}$ with $c_\pi((x,y)) := c((x,y)) + \pi(x) - \pi(y)$ as the the \emph{reduced costs} with respect to $c$ and $\pi$. 

\begin{observation}\label{obs:potentialPaths}
 Let $G=(V,E)$ be a directed graph with edge weights $c: E \rightarrow \mathbb{R}$, let $\pi \colon V \rightarrow \mathbb{R}$ be a function defined on the vertices of $G$, and let $P$ be an $u$-$v$ path in $G$ for $u,v \in V$. Then $c_\pi(P) = c(P) + \pi(u) - \pi(v)$.  
\end{observation}

If for a function $\pi \colon V \rightarrow \mathbb{R}$ it holds that $c_\pi(e) \geq 0$ for all $e \in E$, we call $\pi$ a \emph{feasible potential} for $c$. It is well known that an edge-weighted graph $G$ has no negative cycle if and only if $G$ admits a feasible potential, see e.g.~\cite{Schrijver02}.
To compute a feasible potential for a graph without a negative cycle, one can use the following lemma:

\begin{lemma}[\cite{korteV11}]\label{lem:sppotential}
Let $G = (V,E)$ be a directed graph that has no negative cycles relative to edge weights $c \colon E \rightarrow \mathbb{R}$. Let $x \notin V$ be a new vertex that is connected to all other vertices with edges of weight zero. Then $\pi \colon V \rightarrow \mathbb{R}$ with $\pi(v) = dist_{G,c}(x,v)$ is a feasible potential.
\end{lemma}
We call a feasible potential computed as in Lemma~\ref{lem:sppotential} a \emph{shortest-path feasible potential}.
To define a feasible potential also for vertex-weighted graphs, we consider the edge-shifted weights instead, i.e., where the vertex weight of a vertex is shifted to all outgoing edges. 

\begin{mydefinition}\label{def:EdgeShiftedWeights}
Let $G = (V,E)$ be a directed graph with vertex weights $w \colon V \rightarrow \mathbb{R}$. For any edge $(x,y) \in E$, we define the \emph{edge-shifted weights} by $\shiftedw((x,y)) := w(x)$. 
\end{mydefinition}

Note that for a vertex-weighted graph $G$ with vertex weights $w \colon V(G) \rightarrow \mathbb{R}$, a path $P$ is a shortest path between two vertices $u,v \in V(G)$ w.r.t.~the vertex weights $w$ if and only if it is a shortest path w.r.t.~the edge-shifted weights $\shiftedw$.

\begin{observation}\label{obs:edgeshiftedpath}
Let $G = (V,E)$ be a directed graph with vertex weights $w \colon V \rightarrow \mathbb{R}$, let $P$ be a $u$-$v$ path in $G$ for $u, v \in V$, and let $C$ be a cycle in $G$.  Then  $\shiftedw(P) =w(P) - w(v)$ and $\shiftedw(C) = w(C)$.
\end{observation}

In the following, if we speak about a potential for vertex-weighted graphs we implicitly consider the edge-shifted weights. Alongside to a feasible potential we will compute the value of a minimum shortest path of a subgraph in the algorithm, which we define next.
Note that we consider a single vertex as a path.

\begin{mydefinition}
Let $G = (V,E)$ be a graph and let $w\colon V\to\mathbb{R}$ be vertex weights. We define the \emph{minimum shortest path} value as $msp(G) = \min_{P \in \mathcal{P}} w(P)$ where $\mathcal{P}$ denotes the set of all paths (including single vertices) in $G$.
\end{mydefinition}

In a slight abuse of notation we also use $G$ to stand for a weighted, directed graph, consisting of a directed graph $(V,E)$ and vertex weights $w\colon V\to\mathbb{R}$. For such a vertex-weighted, directed graph $G$, we define $f(G)$ as the function that either returns a negative cycle in $G$ or that returns a shortest-path feasible potential and the minimum shortest path value $msp(G)$. (Formally $f(G):=f(G,w)$ and $msp(G):=msp(G,w)$.)

\subparagraph{Increment.}
Let $G= (V,E)$ be a directed graph and let $w\colon V\to\mathbb{R}$ be vertex weights.
Note that one cannot use Dijkstra's algorithm directly to compute a shortest-path feasible potential since there may be negative weights.
However, in \cite{IwataOO18} it was shown how to compute a feasible potential for $\IncxEx(G)$ in time $\Oh(m + n \log n)$ that is not necessarily a shortest-path feasible potential. This can be used to compute $msp(\IncxEx(G))$ and a shortest-path feasible potential in the same running time.
 
\begin{lemma}\label{lem:NCDInc}
Let $G=(V,E)$ be a directed graph, let $x\notin V$, let $E_x \subseteq \{ (x,u) \mid u \in V\} \cup \{ (u,x) \mid u \in V\}$, let $w\colon V\cup\{x\}\to\mathbb{R}$, and let $f(G)$ be given. Then one can compute $f(\IncxEx(G))$ in time $\Oh(m + n \log n)$ with $n =\vert V \vert$ and $m = \vert E \cup E_x \vert $.
\end{lemma}  

\begin{proof}
If $f(G)$ returns a negative cycle, we do so for $f(\IncxEx(G))$. 
Let otherwise $\shiftedw$ be the edge-shifted weights in $G$.
In \cite{IwataOO18} (Theorem~4),
it was shown how to compute a feasible potential $\pi$ w.r.t.\ $\shiftedw$ for $\IncxEx(G)$ in time $\Oh(m + n \log n)$ if a feasible potential for $G$ is known. 
Let $\shiftedw_\pi \colon E \cup E_x \rightarrow \mathbb{R}$ be the reduced cost with respect to $\shiftedw$ and $\pi$.
Using Dijkstra's algorithm, one can compute the values $dist_{\shiftedw_\pi}(x,v)$ for all $v \in V$ in the same running time and, by reversing all edge directions, also the values $dist_{\shiftedw_\pi}(v,x)$. 
By Observation~\ref{obs:potentialPaths} and Observation~\ref{obs:edgeshiftedpath}, we can then reconstruct the distances w.r.t.\ the edge weights $\shiftedw$ and finally w.r.t.\ the original vertex weights $w$.

A minimum shortest path in $\IncxEx(G)$ either does use the new vertex $x$ or it does not, thus, it holds that $msp(\IncxEx(G))$ is the minimum of the two values $\min_{v \in V} dist(v,x) + \min_{u \in V} dist(x,u) - w(x)$ and $msp(G)$, and can be determined in linear time. 

Finally, the feasible potential $\pi$ that is computed in \cite{IwataOO18} does not need to be a shortest-path feasible potential, however, using the corresponding non-negative edge weights $\shiftedw_\pi$ and Lemma~\ref{lem:sppotential}, the problem of computing a shortest-path feasible potential reduces to the problem of computing a shortest path to all vertices from a new vertex that is connected to every vertex with an edge of weight zero in a graph with \emph{non-negative edge weights}. This can be done, using Dijkstra's algorithm once again, in time $\Oh(m + n \log n)$.
\end{proof}

\subparagraph{Substitution.}
Let $H$ be an unweighted, directed $t$-vertex graph and let $G_i$ be vertex-weighted, directed graphs for $i \in [t]$.
The following lemma is proven in \cite{KratschN20} for positive vertex-weighted graphs, Lemma~\ref{lem:spEntersModuleOnce} generalizes the result to arbitrary vertex weights.

\begin{lemma}[adapted from \cite{KratschN20}]\label{lem:spEntersModuleOnce}
Let $H$ be an unweighted, directed $t$-vertex graph and let $G_i$ be vertex-weighted, directed graphs such that $\SubstH(G_1, \ldots, G_t)$ does not contain a negative cycle. Let further $u, v \in V(\SubstH(G_1, \ldots, G_t))$ such that there exists a shortest $u$-$v$ path and  $\{u,v\} \nsubseteq V(G_i)$ for all $i \in [t]$.
Then there exists a shortest $u$-$v$ path $P$ in $\SubstH(G_1, \ldots, G_t)$ in which the occurrences of vertices in each $G_i$ occur consecutively.
\end{lemma}
\begin{proof}
Denote $G = \SubstH(G_1, \ldots, G_t)$ and let $P = (u=p_1, p_2, \ldots, p_n = v)$ be a shortest $u$-$v$ path in $G$. Assume for contradiction that there exists indices $\alpha, \alpha', \beta, \gamma, \gamma' \in [n]$ with $\alpha \leq \alpha' < \beta < \gamma \leq \gamma'$ such that $p_j \in V(G_i)$ for all $\alpha \leq j \leq \alpha'$ and $\gamma \leq j \leq \gamma'$, but $p_{\beta} \notin V(G_i)$ for some $i \in [t]$. 
We distinguish between the two cases $\alpha \neq 1$ and $\alpha =1$.

If $\alpha \neq 1$ it holds that  $(p_{\alpha-1}, p_{\alpha}) \in E(G)$.
Since $V(G_i)$ is a module in $G$ it also holds that $(p_{\alpha - 1}, p_{\gamma}) \in E(G)$. 
Thus, consider the path $P'=(p_1, \ldots, p_{\alpha - 1}, p_{\gamma}, p_{\gamma+1}, \ldots, p_n)$ and denote by $C' = (p_{\alpha}, \ldots, p_{\gamma-1})$ the skipped path. It holds that $\omega(P) = \omega(P') + \omega(C')$. Since $V(G_i)$ is a module in $G$ with $p_\alpha, p_\gamma \in V(G_i)$ and $(p_{\gamma-1},p_\gamma)\in E(G)$, it holds that also $(p_{\gamma-1}, p_{\alpha}) \in E(G)$. Thus, $C'$ is a cycle and it holds that  $\omega(C') \geq 0$, since $G$ has no negative cycle. Hence, it holds that $\omega(P') \leq \omega(P)$.

If $\alpha = 1$ it holds that $\gamma' \neq n$, since $u$ and $v$ are in different modules and therefore $(p_\gamma, p_{\gamma+1}) \in E(G)$. Since $V(G_i)$ is a module, it also holds that $(p_{\alpha'}, p_{\gamma+1}) \in E(G)$. Thus, consider the path $P'' = (p_1, \ldots, p_{\alpha'}, p_{\gamma'+1}, \ldots, p_n)$ and denote by $C'' = (p_{\alpha'+1}, \ldots, p_{\gamma})$ the skipped path. It holds thats $w(P) = \omega(P'') + \omega(C'')$. 
Since $V(G_i)$ is a module in $G$ with $p_{\alpha'}, p_{\gamma'} \in V(G_i)$ and $(p_{\alpha'},p_{\alpha'+1}) \in E(G)$, it holds that also $(p_{\gamma},p_{\alpha'+1}) \in E(G)$. Thus, $C''$ is a cycle and it holds that $\omega(C'') \geq 0$, yielding again that $\omega(P'') \leq \omega(P)$.

We iterate this procedure for every pair of maximal sequences of vertices that are in a same module but interrupted by at least one vertex of a different module. Since the vertices in $P'$ resp. $P''$ are a strict subset of the vertices in $P$, the number of vertices in the $u$-$v$ path strictly reduces each time.
\end{proof}

Lemma~\ref{lem:spEntersModuleOnce} motivates the following definition.

\begin{mydefinition}\label{def:Homega}
Let $H$ be an unweighted, directed graph with $V(H) = \{v_1, \ldots, v_t\}$ and let $G_i$ be vertex-weighted, directed graphs for $i \in [t]$. Define $\Hw$ as the graph $H$ with vertex weights $w\colon V(H) \rightarrow \mathbb{R}$ defined by $w(v_i) = msp(G_i)$. 
\end{mydefinition}

Let $\Hw$ be the graph defined in Definition~\ref{def:Homega}. 
The next lemma shows that there is a negative cycle in $\SubstH(G_1, \ldots, G_t)$ if and only if there is a negative cycle in $\Hw$ and that the minimum shortest path value in $\SubstH(G_1, \ldots, G_t)$ and $\Hw$ coincide.

\begin{lemma}\label{lem:NegCycleInH}
Let $H$ be an unweighted, directed $t$-vertex graph, let $G_i$ be vertex-weighted, directed graphs without negative cycles, and let $\Hw$ be the graph as defined in Definition~\ref{def:Homega}. Then there exists a negative cycle in $\SubstH(G_1, \ldots, G_t)$ if and only if there exists a negative cycle in  $\Hw$.  Moreover, if $\SubstH(G_1, \ldots, G_t)$ does not admit a negative cycle, $msp(\Hw) = msp(\SubstH(G_1, \ldots, G_t))$.
\end{lemma}

\begin{proof}
Assume that $\SubstH(G_1, \ldots, G_t)$ has a negative cycle. 
Let $C$ be a negative cycle in $\SubstH(G_1, \ldots, G_t)$ of minimum weight. Since no $G_i$ contains a negative cycle, $C$ cannot be completely inside a single $G_i$. Moreover, because $C$ is chosen as a cycle of minimum weight, each subpath of $C$ is a shortest path and thus, by Lemma~\ref{lem:spEntersModuleOnce}, $C$ enters each $G_i$ at most once. 
Denote by $P_i$ the maximal subpaths of $C$ with vertices in $V(G_i)$.
Clearly, $w(P_i) \geq msp(G_i)$, and in fact it even holds that $w(P_i) = msp(G_i)$, since otherwise we could replace $P_i$ by the path $P_i'$ with $w(P_i') = msp(G_i)$, because $V(G_i)$ is a module in $\SubstH(G_1, \ldots, G_t)$, which contradicts the choice of $C$ as a cycle of minimum weight. 
Thus, there is a corresponding cycle $C'$ in $\Hw$ with $w(C) = w(C')$. 
Conversely, each cycle in $\Hw$ directly corresponds to a cycle in $\SubstH(G_1, \ldots, G_t)$ of the same length by replacing each vertex by a corresponding minimum shortest path, hence there is a negative cycle in $\SubstH(G_1, \ldots, G_t)$ if and only if there is one in $\Hw$.

We are left to show $msp(\Hw) = msp(\SubstH(G_1, \ldots, G_t))$ if $\SubstH(G_1, \ldots, G_t)$ does not contain a negative cycle. 
Clearly, $msp(\SubstH(G_1, \ldots, G_t)) \leq msp(\Hw)$, since every path in $\Hw$ corresponds to a path in $\SubstH(G_1, \ldots, G_t)$ of the same length. For the other direction let $P$ be a path in $\SubstH(G_1, \ldots, G_t)$ with $w(P) = msp(\SubstH(G_1, \ldots, G_t))$. 
We distinguish three cases:
(1) $V(P) \subseteq V(G_i)$ for some $i \in [t]$. Then the corresponding vertex $v_i \in V(\Hw)$ has weight at most $w(P)$ and thus $ w(P) \geq msp(G_i) \geq msp(\Hw)$. 
(2) $P$ starts and ends in the same vertex set $V(G_i)$ for some $i \in [t]$. Let $P'$ be the first part of $P$ that is completely in $V(G_i)$ and let $P''$ be the remainder of $P$. Then it holds that $w(P'') \geq 0$ since otherwise there would be a negative cycle in $\SubstH(G_1, \ldots, G_t)$ on the same vertex set of $P''$ ($V(G_i)$ is a module). Thus, $w(P) = w(P') + w(P'') \geq w(P') \geq msp(G_i) \geq msp(\Hw)$. 
(3) $P$ starts in $V(G_i)$ and ends in $V(G_j)$ for $i,j \in [t]$ with $i \neq j$. Then, by Lemma~\ref{lem:spEntersModuleOnce}, $P$ visits each $G_i$ at most once. Let $P^*$ be the path that replaces each maximal subpath of $P$ in some $G_i$ by the minimum shortest path in $G_i$. Now, it holds that $w(P) \geq w(P^*)$ and since there is a corresponding path $P'$ in $\Hw$ with $w(P') = w(P^*)$ we have shown that $  w(P) \geq w(P^*) \geq msp(\Hw)$.
\end{proof}

We can now prove the following lemma:

\begin{lemma}\label{lem:NCDSubst}
Let $H$ be an unweighted, directed $t$-vertex graph, let $G_i$ be vertex-weighted, directed graphs, and let $f(G_i)$ be given for each $G_i$. Then one can compute $f(\SubstH(G_1, \ldots, G_t))$ combinatorially in time $\Oh(t^3 + n)$ or in time $\Oh(t^{2.842} + n)$ using fast matrix multiplication.
\end{lemma}

\begin{proof}
If any $f(G_i)$ indicates a negative cycle, we do so for $f(\SubstH(G_1, \ldots, G_t))$. Let otherwise $\pi_{i}$ be a shortest-path feasible potential for $G_i$ and let $msp(G_i)$ be the minimum shortest path value of $G_i$ for $i \in [t]$.

By Lemma~\ref{lem:NegCycleInH}, it suffices to solve \VWAPSP on $\Hw$ to check if there is a negative cycle in $\SubstH(G_1, \ldots, G_t)$ and (if not) to compute the minimum shortest path value of $\SubstH(G_1, \ldots, G_t)$. This can be done in  time $\Oh(t^{2.842})$ by the algorithm of Yuster~\cite{Yuster09} resp.\ combinatorially in time $\Oh(t^3)$ using a vertex-weighted variant of Floyd's algorithm~\cite{Floyd62}.

We are left to compute a shortest-path feasible potential for $\SubstH(G_1, \ldots, G_t)$. 
To do so, we add a new vertex $x$ to $\Hw$ of weight zero with directed arcs to all other vertices of $\Hw$. This corresponds to adding a new vertex $x'$ to the graph $\SubstH(G_1, \ldots, G_t)$ with directed arcs to all other vertices of $\SubstH(G_1, \ldots, G_t)$.
Now, we compute
a shortest-path feasible potential (regarding the edge-shifted weights) $\pi_H$ for $\Hw$ in time $\Oh(t^{2.842})$ resp.\ $\Oh(t^{3})$.
We define $\pi \colon V(\SubstH(G_1, \ldots, G_t)) \rightarrow \mathbb{R}$ by $\pi(v) = \pi_i(v) + \pi_H(v_i)$ for $v \in G_i$, and claim that $\pi$ is a shortest-path feasible potential for $\SubstH(G_1, \ldots, G_t)$. 
To prove the claim, we show that for any vertex $v \in G_i$ for some $i \in [t]$, a shortest $x'$-$v$ path (regarding the edge-shifted weights) has weight $\pi_i(v) + \pi_H(v_i)$.
Since $x'$ is connected to every vertex in $V(\SubstH(G_1, \ldots, G_t))$, the singleton set $\{x'\}$ forms a module and thus we know by Lemma~\ref{lem:spEntersModuleOnce} that there exists a shortest $x'$-$v$ path that does not enter any $G_i$ twice. Thus, a shortest path from $x'$ to $v$ can be split into a shortest path to reach $V(G_i)$ and a shortest path from some vertex in $G_i$ to $v$. The latter part is exactly $\pi_i(v)$ and the former part can be constructed by using a minimum shortest path in each $G_j$ for $j \in [t]$, i.e., is equal to $\pi_H(v_i)$.    
\end{proof}

\begin{lemma}\label{lem:NCDSubstTd}
Let $H$ be an unweighted, directed $t$-vertex graph, let $G_i$ be vertex-weighted, directed graphs, and let $f(G_i)$ be given for each $G_i$. Then one can compute $f(\SubstH(G_1, .., G_t))$ in time $\Oh(\td(H) \cdot (\vert E(H) \vert + t\log t) + n)$ with $n = \vert \SubstH(G_1, .., G_t)) \vert$.
\end{lemma}

\begin{proof}
If any $f(G_i)$ indicates a negative cycle, we do so for $f(\SubstH(G_1, \ldots, G_t))$. Otherwise, by Lemma~\ref{lem:NegCycleInH}, $\SubstH(G_1, \ldots, G_t)$ has a negative cycle if and only if $\Hw$ does and it holds that $msp(\SubstH(G_1, \ldots, G_t)) = msp(\Hw)$. Moreover, as seen in the proof of Lemma~\ref{lem:NCDSubst}, one can compute a shortest-path feasible potential for $\SubstH(G_1, \ldots, G_t)$ in linear time after we have computed a shortest-path feasible potential for $\Hw$. Thus, we can exploit the tree-depth expression of $\Hw$ by using Lemma~\ref{lem:NCDInc}. Using the tree-depth running-time framework by Iwata et al.~\cite{IwataOO18}, we have proven the lemma.
\end{proof}

\begin{proof}[Proof of Theorem~\ref{thm:NegCycleDetection}.]
Let $\expr$ be the given expression with $val(\expr) = G$. We use Theorem~\ref{thm:runningtime} to prove the claim. Lemma~\ref{lem:NCDInc}, Lemma~\ref{lem:NCDSubst}, resp.\ Lemma~\ref{lem:NCDSubstTd} provide the algorithms $A_{\mathrm{Inc}}$, $A_{\mathrm{Sub}}$, resp.\ $A_{\mathrm{SubTd}}$. Since $\vert \T{\expr} \vert \leq n$, the linear portions of the running time of $A_{\mathrm{Sub}}$ and $A_{\mathrm{SubTd}}$ sum up to $\Oh(n^2)$. Moreover, $f(\circ)$ and $f(\bullet)$ can be computed in constant time. 
Thus, by Theorem~\ref{thm:runningtime}, the claim follows. 
\end{proof}

We emphasize that for the algorithm $A_{\mathrm{Inc}}$ only a feasible potential is needed and for the algorithm $A_{\mathrm{Sub}}$ the value of a minimum shortest path alone would suffice, but in the heterogeneous case we need to compute both. 
We consider a \emph{shortest-path} feasible potential to ease the computation of a (shortest-path) feasible potential in $A_{\mathrm{Sub}}$ and $A_{\mathrm{SubTD}}$.

\subsection{Vertex-Weighted All-Pairs Shortest Paths}\label{section:vwapsp}

In this section we will extend the algorithm of Section~\ref{sec:NegCycle} and compute for all pairs of vertices the shortest-path distance. We will prove the following Theorem:

\begin{theorem}\label{thm:VWAPSP}
Let $G = (V,E)$ be a directed graph such $G \in \TdkMwhMtdl$ with a given expression for some $k,h,\ell \in \N$, and let $w \colon V \rightarrow \mathbb{R}$. Then one can either conclude that $G$ contains a negative cycle or one can solve \VWAPSP combinatorially in time $\Oh(h^2 n + (k+\ell) n^2)$ resp.\ in time $\Oh(h^{1.842} n + (k+\ell) n^2)$ using fast matrix multiplication.
\end{theorem}

We will compute slightly different values $f(G)$ depending on the operation.
Most notably, we will compute the shortest path values between all pair of vertices only before and after an operation $\IncxEx$, whereas for operations $\SubstH$, we restrict the  computation to auxiliary values between modules and single vertices.

\subparagraph{Substitution.}

For a directed, unweighted $t$-vertex graph $H$ with vertex set $V(H) = \{v_1, \ldots, v_t\}$ and directed, vertex-weighted graphs $G_i$ let $\Hw$ be the graph defined in Definition~\ref{def:Homega} and let $G = \SubstH(G_1, \ldots, G_t)$.

For $u \in V(G_i)$ and $v \in V(G_j)$ with $i \neq j$, a shortest $u$-$v$ path in $G$ will consists of a shortest path from $u$ to some vertex in $V(G_i)$ (possibly only $u$), followed by a shortest path from $v_i$ to $v_j$ in $\Hw$ (i.e., using the minimum shortest path in all intermediate modules\footnote{Hence, we do not consider the weights of the start- and endvertex.}), completed by a shortest path from some vertex in $V(G_j)$ to $v$ (possibly only $v$).
This motivates the following definition:
We define the function $f_S(G)$ as the function that returns the same values that we have computed in Section~\ref{sec:NegCycle}, i.e., a shortest-path feasible potential and the value $msp(G)$ of a minimum shortest path in $G$. Additionally, $f_S(G)$ returns for each pair of vertices $v_i,v_j \in V(\Hw)$ the length of a shortest $v_i$-$v_j$ path in $\Hw$ for $i, j \in [t]$, and for each vertex $u \in V(G)$ the values $\min_{v \in V(G)}dist_G(u,v)$ and $\min_{v \in V(G)}dist_G(v,u)$. See also the left side of Table~\ref{tab:valuesfSfI} for the list of values returned by the function $f_S$.

\begin{lemma}\label{lem:VWAPAPSubst}
Let $H$ be an unweighted, directed $t$-vertex graph, let $G_i$ be vertex-weighted, directed graphs, and let $f_S(G_i)$ be given for each $G_i$. Then one can compute $f_S(\SubstH(G_1, .., G_t))$ in $\Oh(t^3 + n)$ combinatorially time or $\Oh(t^{2.842} + n)$ time using fast matrix multiplication.
\end{lemma}
\begin{proof}
If any $f(G_i)$ indicates a negative cycle, we do so for $f_S(\SubstH(G_1, .., G_t))$. 
Otherwise, let $\Hw$ be as defined in Definition~\ref{def:Homega} and denote $G = \SubstH(G_1, \ldots, G_t)$.
First of all, we solve the vertex-weighted all-pairs shortest paths problem on $H_w$, check for a negative cycle, compute a shortest-path feasible potential, and compute the minimum shortest path value $msp(G)$ as done in Section~\ref{sec:NegCycle} for \NCD. We are left to compute for each vertex $u \in V(G)$ the values $\min_{v \in V(G)}dist_G(u,v)$ and $\min_{v \in V(G)}dist_G(v,u)$. 

Let $i \in [t]$ such that $u \in V(G_i)$. If $\argmin_{v \in V(G)}dist_G(u,v) \notin V(G_i)$ we know due to Lemma~\ref{lem:spEntersModuleOnce} that there exists a shortest $u$-$v$ path in $G$ in which the occurences of vertices in each $G_i$ occur consecutively. Thus, in this case it holds that 
\begin{align}
\min_{v \in V(G)}dist_G(u,v) = \min_{u' \in V(G_i)}dist(u,u') + \min_{v_j \in V(\Hw)}dist_{\Hw}(v_i, v_j) - w(v_i).  \label{align:mindistfromu}  
\end{align}

If $\argmin_{v \in V(G)}dist_G(u,v) \in V(G_i)$, we observe that in this case it holds that the length of a shortest $u$-$v$ path is equal to  $\min_{u' \in V(G_i)}dist(u,u')$; otherwise there would be a negative cycle in $G$, since w.l.o.g.\ one can assume that each shortest $u$-$v$ path does start with a path in $V(G_i)$ of length $\min_{u' \in V(G_i)}dist(u,u')$. Thus, equation (\ref{align:mindistfromu}) does hold in general.
To compute the value in (\ref{align:mindistfromu}) for all $u \in V(G)$, we first determine for each $v_i \in V(\Hw)$ the value $\min_{v_j \in V(\Hw)}dist_{\Hw}(v_i, v_j)$ in time $\Oh(t^2)$ and store them.
Since the value of $\min_{u' \in V(G_i)}dist(u,u')$ is known by $f(G_i)$, we can compute the value in (\ref{align:mindistfromu}) for each $u \in V(G)$ in constant time.

The value $\min_{v \in V(G)}dist_G(v,u)$ can be computed analogously by executing the whole algorithm at any time also for the edge-flipped graph.
\end{proof}

\begin{table}[t]
 \begin{tabular}{l|l}
  \multicolumn{1}{c|}{Values returned by $f_S(G)$}& \multicolumn{1}{c}{Values returned by $f_I(G)$}\\
  \hline
  &\\[-1.1em]
  $\bullet$ shortest-path feasible potential $\pi$     &$\bullet$  shortest-path feasible potential $\pi$\\
  $\bullet$ $msp(G) = \min_{u,v \in V(G)}dist_G(u,v)$     &$\bullet$  $msp(G) = \min_{u,v \in V(G)}dist_G(u,v)$\\
  $\bullet$ $\min_{v \in V(G)} dist_G(u,v)$ for all $u \in V(G)$    &$\bullet$  $\min_{v \in V(G)} dist_G(u,v)$ for all $u \in V(G)$\\
  $\bullet$ $\min_{v \in V(G)} dist_G(v,u)$ for all $u \in V(G)$    &$\bullet$  $\min_{v \in V(G)} dist_G(v,u)$ for all $u \in V(G)$\\
  $\bullet$ $dist_{H_\omega}(v_i, v_j)$ for all $v_i, v_j \in V(H_\omega)$ &$\bullet$  $dist_G(u,v)$ for all $u,v \in V(G)$
 \end{tabular}
 \caption{An Overview of the values returned by the functions $f_S$ resp.\ $f_I$ for the problem $\VWAPSP$. All but the last bullet are identical. For the problem \NCD the first two bullets are sufficient.}
 \label{tab:valuesfSfI}
\end{table}

\subparagraph{Increment.}
For the operation $\IncxEx$, we define $f_I(G)$ as a function that returns the same values as $f_S(G)$ but instead of the pairwise distances in the graph $H_\omega$ it returns the pairwise distance for all pairs $u,v \in V(G)$. See  also Table~\ref{tab:valuesfSfI} for an summary of the values returned by the function $f_I$.
We compute $f_I(G)$ before and after the $\Inc$-operation. 

Recall that for an expression $\expr$, as defined in Definition~\ref{definition:furtherclasses}, we denote by $\T{\expr}$ the corresponding expression tree and for a node $r\in V(\T{\expr})$ we denote by $\Tv{\expr}{r}$ the subexpression tree of $\T{\expr}$ with root $r$. Further, we will denote by $\Gv{\expr}{r} = val(\Tv{\expr}{r})$ the resulting graph after evaluating $\Tv{\expr}{r}$.
Consider a node $r \in V(\T{\expr})$ labeled according to an operation $\SubstH$ whose parent node is labeled according to an operation $\IncxEx$.

To compute $f_I(\Gv{\expr}{r})$ we will consider the maximal subtree of $\Tv{\expr}{r}$ that only admits nodes labeled according to an operation $\SubstH$. 
Considering this subexpression tree, one can prove in a similar way as done in \cite{KratschN20} the following lemma:

\begin{lemma}\label{lem:fStofI}
Let $\T{\expr}$ be an expression tree of an expression $\expr$ as defined in Definition~\ref{definition:furtherclasses}, let $r \in V(\T{\expr})$ labeled according to an operation $\SubstH$ whose parent node is labeled according to an operation $\IncxEx$, and let $\Gv{\expr}{r}$ be the resulting graph after evaluating $\Tv{\expr}{r}$.
Let further the function $f_S$ resp.\ $f_I$ be known for all graphs corresponding to nodes in $V(\Tv{\expr}{r})$ labeled according to an operation $\SubstH$ resp.\ $\IncxEx$. Then, one can compute $f_I(\Gv{\expr}{r})$ in time $\Oh(n^2)$ with $n = \vert V(\Gv{\expr}{r}) \vert$.
\end{lemma}

\begin{proof}
Since $f_S(\Gv{\expr}{r})$ is known, we are left to compute $dist_{\Gv{\expr}{r}}(u,v)$ for all $u,v \in V(\Gv{\expr}{r})$ in order to compute $f_I(\Gv{\expr}{r})$. 
For this, we traverse the expression tree of $\Tv{\expr}{r}$ downwards as long as one encounters a node that is not labeled according to an operation $\Subst_H$. W.l.o.g.\ we can assume that such a node is either labeled according to an operation $\IncxEx$ or $\bullet$ (a single vertex), since one can assume that no argument of an operation $\SubstH$ is the empty graph and one can replace the operations $\Union$ and $\Jn$ by operations $\SubstH$ with pattern graphs of size two, as seen in the proof of Theorem~\ref{thm:runningtime}. 
We denote the vertex set of this part of the expression tree (including the first nodes that are not labeled according to an operation $\SubstH$) by $\U_r$.
Note that for a vertex $x \in \U_r$ the corresponding vertex set $V(\Gv{\expr}{x})$ forms a module in $\Gv{\expr}{r}$.
We will determine for each such $\Gv{\expr}{x}$ the value of a shortest $u$-$v$ path in $\Gv{\expr}{r}$ only using vertices in $V(\Gv{\expr}{r}) \setminus V(\Gv{\expr}{x})$ with the property that $u$ and $v$ are adjacent to (all vertices of) $V(\Gv{\expr}{x})$. We denote this value by $c_x$ and we set $c_r = \infty$.

Assume for now that one has computed the values $c_x$ for all $x \in \U_r$ in time $\Oh(n^2)$ for $n = \vert V(\Gv{\expr}{r}) \vert$.
Consider a node $x \in \U_r$ labeled according to an operation $\SubstH$ and let $\Gv{\expr}{x} = \SubstH(G_1, \ldots, G_t)$ be the corresponding graph with each $G_i$ corresponds to a child of $x$. Let $H_w$ be the graph defined in Definition~\ref{def:Homega}.

For two vertices $u,v$ with $u \in V(G_i)$ and $v \in V(G_j)$ for $i \neq j$, it holds that a shortest $u$-$v$ path in $\Gv{\expr}{r}$ is either completely in $\Gv{\expr}{x}$ or it does use vertices in $V(\Gv{\expr}{r})\setminus V(\Gv{\expr}{x})$.
Since there is no negative cycle in $\Gv{\expr}{r}$ and due to Lemma~\ref{lem:spEntersModuleOnce}, a shortest $u$-$v$ path completely in $\Gv{\expr}{x}$ has length 
$\min_{u' \in V(G_i)}dist_{G_i}(u,u') + dist_{\Hw}(v_i, v_j) - w(v_i) - w(v_j) + \min_{v' \in V(G_j)}dist_{G_j}(v',v)$. 
If a shortest $u$-$v$ path does use vertices in $V(\Gv{\expr}{r})\setminus V(\Gv{\expr}{x})$, we can determine the length by $\min_{u' \in V(\Gv{\expr}{x})}dist_{\Gv{\expr}{x}}(u,u') + c_x + \min_{v' \in V(\Gv{\expr}{x})}dist_{\Gv{\expr}{x}}(v',v)$.
Thus, for a node $x \in \U_r$ labeled according to an operation $\SubstH$, we can compute the shortest-path distance of two vertices that are in different modules by taking the minimum of those two values.

For a node $x \in \U_r$ that is not labeled according to an operation $\SubstH$, we have already computed $f_I(\Gv{\expr}{x})$ and thus, $dist_{\Gv{\expr}{x}}(u,v)$ is known for all $u,v \in V(\Gv{\expr}{x})$ and it holds that $dist_{\Gv{\expr}{r}}(u,v) = \min\{dist_{\Gv{\expr}{x}}(u,v), \min_{u' \in V(\Gv{\expr}{x})}dist_{\Gv{\expr}{x}}(u,u') + c_x + \min_{v' \in V(\Gv{\expr}{x})}dist_{\Gv{\expr}{x}}(v',v) \}$.

Note that for a pair of vertices $u,v \in V(\Gv{\expr}{r})$ it either holds that $u,v \in V(\Gv{\expr}{x})$ for some $x \in \U_r$ that is not labeled according to an operation $\SubstH$ or it holds that $u,v \in V(\Gv{\expr}{x})$ for some $x \in \U_r$ that is labeled according to an operation $\SubstH$ and $u$ and $v$ are in different modules. 
Since all the considered values are known, we can compute $dist_{G_r}(u,v)$ for all $u,v \in V(G_r)$ in time $\Oh(n^2)$.

We are left to compute the values $c_x$ for each node $x \in \U_r$.
We do this in a top-down traversal of the nodes in $\U_r$. For the root $r \in \U_r$ it holds that $c_r =  \infty$.
Let $x \in \U_r$ be a node labeled according to an operation $\SubstH$ and let $v_1, \ldots, v_t \in \Tv{\expr}{r}$ be the children of $x$, i.e., $\Gv{\expr}{x} = \SubstH(\Gv{\expr}{v_1}, \ldots, \Gv{\expr}{v_t})$. Further, let $\Hw$ as defined in Definition~\ref{def:Homega} with vertex set $\{v_1, \ldots, v_t\}$ and vertex weights $\omega$. Inductively, we can assume that $c_x$ is know.
Now, $c_{v_i}$ corresponds to a cycle that either only uses vertices in $V(\Gv{\expr}{x})$ or it uses vertices in $V(\Gv{\expr}{r}) \setminus V(\Gv{\expr}{x})$. Thus, $c_{v_i}$ is the minimum of the two values: 
\begin{itemize}
\item $\min_{C \in \mathcal{C}_{v_i}}\omega(C) - \omega(v_i)$ with $\mathcal{C}_{v_i}$ is the set of all cycles from $v_i$ to $v_i$ in $\Hw$
\item $\min_{v_j \in V(\Hw)} dist_{\Hw}(v_i, v_j) - \omega(v_i) + c_{x} + \min_{v_j \in V(\Hw)}(v_j ,v_i) - \omega(v_i)$.
\end{itemize} 
Since all considered values are known, we can compute the values $c_x$ for each $x \in U_r$ by a top-down traversal in time $\Oh(\vert \U_r \vert) \subseteq \Oh(n)$.
\end{proof}

\begin{lemma}\label{lem:VWAPSPInc}
Let $G' = (V',E')$ be a directed graph, let $x \notin V(G')$ , let $E_x \subseteq \{ (x,u) \mid u \in V'\} \cup \{ (u,x) \mid u \in V'\}$, let $w \colon V' \cup \{x\} \rightarrow \mathbb{R}$, and let $f_I(G')$ be given. Then one can compute $f_I(\IncxEx(G'))$ in time $\Oh(n^2)$, with $n = \vert V' \vert$.
\end{lemma}
\begin{proof}
Let $G = \IncxEx(G')$.
By Lemma~\ref{lem:NCDInc}, we can compute the value $msp(G)$ of a minimum shortest path and a feasible potential for $G$ in time $\Oh(m + n \log n)$. Apply Dijkstra's algorithm twice to compute the values $dist_G(x,v)$ and $dist_G(v,x)$ for all $v \in V'$ in the same running time. For $u,v \in V'$, a shortest $u$-$v$ path in $G$ does either use the vertex $x$, or does not use the vertex $x$. Thus, we can update the shortest-path distance for each pair $u,v \in V'$ by $ dist_G(u,v) = \min \{dist_{G'}(u,v), dist_G(u,x) + dist_G(x,u)\}$, which can be done for each pair in constant time. Additionally, for each $v \in V(G)$, the values $\min_{u \in V(G)} dist(v,u)$ and $\min_{u \in V(G)} dist(u,v)$ can be looked up in linear time, yielding a total running time of $\Oh(n^2)$ to compute $f_I(G)$.
\end{proof}

\begin{lemma}\label{lem:VWAPSPSubstTd}
Let $H$ be an unweighted, directed $t$-vertex graph, $G_i$ be vertex-weighted, directed graphs, and $f_S(G_i)$ be given for each $G_i$. Then one can compute $f_S(\SubstH(G_1, \ldots, G_t))$ in time $\Oh(\td(H) \cdot t^2 + n)$ with $n = \vert V(\SubstH(G_1, \ldots, G_t)) \vert$.
\end{lemma}
\begin{proof}

Denote $G = \SubstH(G_1, \ldots, G_t)$. Consider the graph $H_\omega$ as defined in Definition~\ref{def:Homega} with $V(H_\omega) = \{v_1, \ldots, v_t\}$. Note that $\td(H_\omega)=\td(H)$, i.e., $H_\omega \in \Td{\td(H)}$, thus, by Lemma~\ref{lem:VWAPSPInc} and Corollary~\ref{cor:runningtimes}, one can compute $f_I(H_\omega)$, especially \VWAPSP, in $H_\omega$ in time $\Oh(\td(H) \cdot t^2)$. Since $f_S(G_i)$ is known, one can compute the values $\min_{v \in V(G)}dist_G(u,v)$ resp. $\min_{v \in V(G)}dist_G(v,u)$ as done in the proof of Lemma~\ref{lem:VWAPAPSubst} in time $\Oh(n)$. Moreover, a shortest-path feasible potential for $G$ and the value $msp(G)$ can also be computed as done in Section~\ref{sec:NegCycle} for \NCD in $\Oh(t^2)$.  
\end{proof}

\begin{proof}[Proof of Theorem~\ref{thm:VWAPSP}]
Let $\expr$ be the given expression after replacing each occurrence of an operation $\Jn$ or $\Union$ by (possible multiple) operations $\SubstH$ with a pattern graph $H$ of size two. We traverse the expression tree $\T{\expr}$ from bottom to top. The values $f_I(\bullet)$,$f_S(\bullet)$, $f_I(\circ)$, and $f_I(\circ)$ can be trivially computed in constant time. For a node $x\in \T{\expr}$ labeled according to an operation $\SubstH$, the algorithms $A_\mathrm{Sub}$ resp.\ $A_{SubTd}$ are as described in Lemma~\ref{lem:VWAPAPSubst} resp.\ Lemma~\ref{lem:VWAPSPSubstTd} and compute $f_S(\Tv{\expr}{x})$. For each node $x \in \T{\expr}$ labeled according to an operation $\IncxEx$, let $y \in \T{\expr}$ be the unique child of $x$. 
If $y$ is labeled according to an operation $\SubstH$, the algorithm $A_{\mathrm{Inc}}$ first computes $f_I(\Tv{\expr}{y})$ using Lemma~\ref{lem:fStofI}. Then, the algorithm computes $f_I(\Tv{\expr}{x})$ using Lemma~\ref{lem:VWAPSPInc}.
By Theorem~\ref{thm:runningtime}, we have proven the theorem.
\end{proof}

\section{Comparing the graph classes}\label{section:relations}

In this section we compare the graph classes defined in Section~\ref{section:heterogeneousstructure}.
It is well known that the graph classes $\Tdk$ and $\Mwh$ are incomparable: 
For any $k,h \in \N$ it holds that $K_{k+1} \in \Mw{0}$, but $K_{k+1} \notin \Tdk$, for $K_{k+1}$ being the clique of size $k+1$. Conversely, let $G$ be a subdivided star of degree $h$, i.e, the graph $K_{1,h}$ with a pendant vertex attached to each vertex of degree one, cf.\ Figure~\ref{Fig:ExampleGraph}. Then it holds that $G \in \Td{3}$, but $G \notin \Mwh$.

Naturally it holds that $\Mw{h} \subseteq \TdMw{0}{h}$ and $\Td{k} \subseteq \TdMw{k}{0}$ for any $k,h \geq 0$; in fact, it can be observed that $\Mwh=\TdMw{0}{h}$ and $\Td{k}\subsetneq\TdMw{k}{0}$.\footnote{Similar trivial relations are true for other classes, we do not state them explicitly.}
Conversely, the following lemma shows that for any $k,h\in\N$ there exist graphs that are neither in $\Tdk$ nor in $\Mwh$, but that are contained in $\TdMw{k'}{h'}$ even for constant $k'$ and $h'$. 

\begin{lemma}\label{lem:TDMWvsTDandMW}
For any $k, h \in \N$, there exists a graph $G$ such that $G \notin \Td{k}$ and $G \notin \Mwh$, but $G \in \TdMw{1}{0}$.
\end{lemma}
\begin{proof}
 Let $p = \max\{k, h, 2\}$. We construct a graph $G$ with $2p+1$ vertices as follows: Consider $p$ pairs of non-adjacent vertices $v_{i,1}$ and $v_{i,2}$ for $i \in [p]$ and one additional vertex called $x$, i.e.,  
 $V = \{v_{1,1}, v_{1,2}, v_{2,1}, v_{2,2}, \ldots, v_{p,1}, v_{p,2}, x\}$.
 There is an edge between any $v_{i,r}$ and $v_{j,s}$ if and only if $i \neq j$ for $r,s \in \{1,2\}$.
 The vertex $x$ is connected to each $v_{i,1}$ for $i \in [p]$.
 See also Figure~\ref{Fig:ExampleGraph}.
 Now, in this constructed graph $G$ the vertex set $\{v_{i,1} \mid i \in [p] \}$ (as the set $\{v_{i,2} \mid i \in [p] \}$) is a clique of size $p$, thus, $\td(G) \geq p$. On the other side, the graph $G$ does not admit any non-trivial module, since it is easy to see that every minimal extension of a two-vertex set results in the whole vertex set $V(G)$. Thus, $\mw(G) = \vert V(G) \vert = 2p+1$.
 
 At the same time, $G \setminus \{x\} \in \TdMw{0}{0}$ since $G \setminus \{x\}$ is a cograph, implying that $G \in \TdMw{1}{0}$.
\end{proof}

Clearly, also the notion of modular tree-depth generalizes both modular-width and tree-depth:
For a graph $G$, the modular tree-depth matches the tree-depth for the special case of $G$ being prime and since it holds that $\td(G) \leq \vert V(G) \vert$ for any graph $G$, it also holds that $\mtd(G) \leq \mw(G)$.
The next lemma shows that in general neither the modular-width nor the tree-depth can be bounded by a function depending on the modular-tree-depth of a graph. 

\begin{lemma}\label{lem:MtdMoreGeneralAsTdMw}
For any $k,h\in\N$, there exists a graph $G$ such that $G \notin \Tdk$ and $G \notin \Mwh$, but $G \in \Mtd{3}$.
\end{lemma}
\begin{proof}
Let $p \in \mathbb{N}$ such that $2 p + 1 > h$. We construct a pattern graph $H$ with $2 p + 1$ vertices as follows: Let $H$ be a subdivided star of degree $p$ with $2p + 1$ vertices, cf.\ Figure~\ref{Fig:ExampleGraph}. Denote $V(H) = \{v_1, v_2, \ldots, v_{2p+1}\}$, and let $G = \SubstH(K_{k+1},\ldots, K_{k+1})$.
Now, $G \notin \Td{k}$ since $G$ contains an clique of size $k+1$ and $G \notin \Mw{h}$, since $G$ contains an induced subdivided star with more than $h$ vertices.

At the same time, since $\td(H) = 3$ and $K_{k+1} \in \Mtd{0}$, it holds that $G\in\Mtd{3}$.
\end{proof}

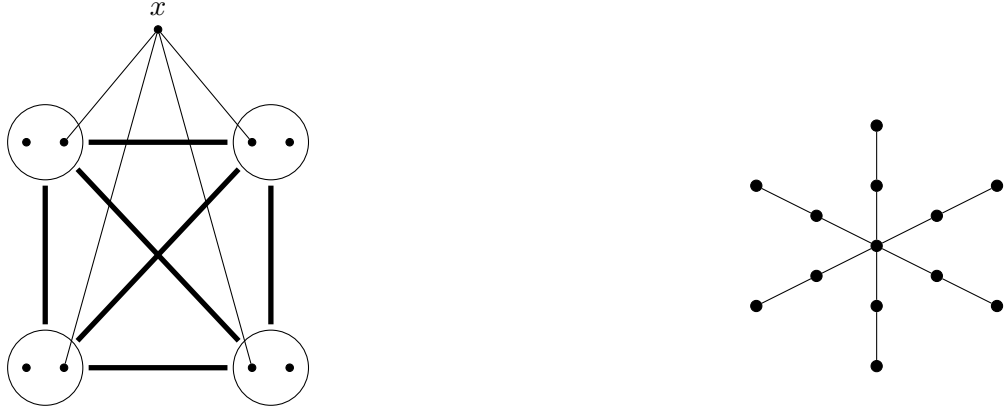
\begin{figure}[t]
\centering
\begin{tikzpicture}
   \def \gap {2}
   \def \distance {3}
   \node (dummy) at (-2,0) {};
   \node[draw, circle, minimum size = 1cm] (dl) at (0,0) {};
   \node[draw, circle, minimum size = 1cm] (ul) at (0,\distance) {};
   \node[draw, circle, minimum size = 1cm] (ur) at (\distance,\distance) {};
   \node[draw, circle, minimum size = 1cm] (dr) at (\distance,0) {};
   \node[draw, circle,fill, inner sep = 1pt] (dlr) at (0.25,0) {};
   \node[draw, circle,fill, inner sep = 1pt] (dll) at (-0.25,0) {};  
   \node[draw, circle,fill, inner sep = 1pt] (drl) at (\distance - 0.25,0) {};
   \node[draw, circle,fill, inner sep = 1pt] (drr) at (\distance + 0.25,0) {}; 
   \node[draw, circle,fill, inner sep = 1pt] (url) at (\distance - 0.25,\distance) {};
   \node[draw, circle,fill, inner sep = 1pt] (urr) at (\distance + 0.25,\distance) {}; 
    \node[draw, circle,fill, inner sep = 1pt] (ull) at (- 0.25,\distance) {};
    \node[draw, circle,fill, inner sep = 1pt] (ulr) at ( 0.25,\distance) {};
     
    \node[draw, circle,fill, inner sep = 1pt, label = {$x$}] (x) at (\distance/2,\distance + \distance/2) {}; 
    \draw[>=latex,,line width=2.000] ([yshift=\gap] dl.north) to ([yshift=-\gap] ul.south);
    \draw[>=latex,line width=2.000] ([xshift=\gap] ul.east) to ([xshift=-\gap] ur.west);
    \draw[>=latex,line width=2.000] ([yshift=-\gap] ur.south) to ([yshift=\gap] dr.north);
    \draw[>=latex,line width=2.000] ([xshift=-\gap] dr.west) to ([xshift=\gap] dl.east);
    \draw[>=latex,line width=2.000] ([xshift=-\gap] dr.north west) to ([xshift=\gap] ul.south east);
    \draw[>=latex,line width=2.000] ([xshift=-\gap] ur.south west) to ([xshift=\gap] dl.north east);  
    \draw[>=latex] (x) to (url);
    \draw[>=latex] (x) to (ulr);
    \draw[>=latex] (x) to (drl);
    \draw[>=latex] (x) to (dlr);
\end{tikzpicture}
\hfill
\begin{tikzpicture}[scale = 0.2]
	    \node  (dummy) at (16,-10) {};
		\node [draw, circle,fill, inner sep = 1.5pt] (0) at (0, 0) {};
		\node [draw, circle,fill, inner sep = 1.5pt] (1) at (-4, 2) {};
		\node [draw, circle,fill, inner sep = 1.5pt] (2) at (-8, 4) {};
		\node [draw, circle,fill, inner sep = 1.5pt] (3) at (0, 4) {};
		\node [draw, circle,fill, inner sep = 1.5pt] (4) at (0, 8) {};
		\node [draw, circle,fill, inner sep = 1.5pt] (5) at (0, -4) {};
		\node [draw, circle,fill, inner sep = 1.5pt] (6) at (0, -8) {};
		\node [draw, circle,fill, inner sep = 1.5pt] (7) at (-4, -2) {};
		\node [draw, circle,fill, inner sep = 1.5pt] (8) at (-8, -4) {};
		\node [draw, circle,fill, inner sep = 1.5pt] (9) at (4, 2) {};
		\node [draw, circle,fill, inner sep = 1.5pt] (10) at (8, 4) {};
		\node [draw, circle,fill, inner sep = 1.5pt] (11) at (4, -2) {};
		\node [draw, circle,fill, inner sep = 1.5pt] (12) at (8, -4) {};

		\draw[] (8) to (7);
		\draw[] (7) to (0);
		\draw[] (5) to (0);
		\draw[] (5) to (6);
		\draw[] (0) to (3);
		\draw[] (3) to (4);
		\draw[] (1) to (0);
		\draw[] (1) to (2);
		\draw[] (0) to (9);
		\draw[] (9) to (10);
		\draw[] (0) to (11);
		\draw[] (11) to (12);	
\end{tikzpicture}
\caption{Left: The graph constructed in the proof of Lemma~\ref{lem:TDMWvsTDandMW} for $p = 4$. The bold edges indicate a full join. Right: A subdivided star of degree six.}
\label{Fig:ExampleGraph}
\end{figure}

As a consequence of Lemma~\ref{lem:TDMWvsTDandMW} and Lemma~\ref{lem:MtdMoreGeneralAsTdMw} it follows that also the classes $\MwhMtdl$ and $\TdkMtdl$ are strictly more general then the classes $\Mwh$ and $\Tdk$.

Next, we will compare the class $\Mtdl$ from Definition~\ref{definition:modulartreedepth} with the class $\TdkMwh$ from Definition~\ref{definition:tdkmwh}. We will see that they are indeed incomparable. A simple extension of the proof of Lemma~\ref{lem:TDMWvsTDandMW} yields the following lemma.

\begin{lemma}\label{lem:TDMWsmallMTDlarge}
 For any $\ell\in\N$, there exists a graph $G$ such that $G \notin \Mtdl$, but $G \in \TdMw{1}{0}$
\end{lemma}
\begin{proof}
 Let $G$ be the graph constructed in the proof of Lemma~\ref{lem:TDMWvsTDandMW} with $p = \ell$, cf.\ left side of Figure~\ref{Fig:ExampleGraph} left. As mentioned in the proof of Lemma~\ref{lem:TDMWvsTDandMW},
 it holds that $G \in \TdMw{1}{0}$. However, the constructed graph $G$ is prime and $\td(G) \geq p$. Thus, $\mtd(G) \geq p$. 
\end{proof}

Conversely, for each $k,h\in\N$, already the graph class $\Mtd{3}$ contains graphs that are not in $\TdkMwh$.

\begin{lemma}\label{lem:MTDandTDMWincomparable}
For any $k, h \in\N$, there exists a graph $G$ such that $G \notin \TdMw{k}{h}$, but $G \in \Mtd{3}$.
\end{lemma}

\begin{proof}
Let $p \in \mathbb{N}$ such that $2 p + 1 > h$. We construct $G$ as done in Lemma~\ref{lem:MtdMoreGeneralAsTdMw}, i.e., let $H$ be a subdivided star with $2 p + 1$ vertices with $V(H) = \{v_1, \ldots, v_{2p+1}\}$, and let $G = \SubstH(K_{k+1}, \ldots, K_{k+1})$.  Denote by $M_i$, for $i \in [2p+1]$, the $2p+1$ many modules in $G$. It was already shown that $G \in \Mtd{3}$.

We claim that $G \notin \TdkMwh$: 
Due to the structure of $G$, for any two vertices $u \in M_i$, $v \in M_j$, with $i \neq j$, it holds that there is no module in $G$ containing $u$ and $v$. Thus, any modular partition of $G$ consists of at least $2p+1$ many modules, which implies that $G$ cannot be in $\TdMw{0}{j}$ for $j \leq h < 2p+1$. 
Moreover, this statement holds even after deleting any $k$ vertices in $G$. Thus, $G$ cannot be in $\TdMw{i}{j}$ for any $i \leq k$ and $j \leq h$.
\end{proof}

Although the modular tree-depth is always upper bounded by the modular-width, it is possible that these two parameters only differ by a constant factor. Thus, depending on the application and the input, it may be beneficial to consider the modular-width instead of the modular tree-depth, since its (asymptotic) impact on the running time may be much lower.

\begin{lemma}\label{lem:MtdlToMw2l}
For any $\ell \in\N_{\geq 2}$, there exists a graph $G$ such that $G \notin \Mtdl$, but $G \in \Mw{2\ell}$.
\end{lemma}

\begin{proof}
Consider a graph $G$ that is a clique $K_{\ell}$ of size $\ell$ with a pendant vertex for each vertex in $K_\ell$. Now, there is no non-trivial module in $G$ and since $G$ contains a clique of size $\ell$ it holds that $\td(G) > \ell$, thus, $G \notin \Mtdl$. However, since $\vert V(G) \vert = 2\ell$ it trivially holds that $\mw(G) = 2\ell$.
\end{proof}

Next, we show that the classes $\TdkMwh$ and $\MwhMtdl$ are incomparable.

\begin{lemma}
For any $h, \ell \in\N$, there exists a graph $G$ such that $G \notin \MwhMtdl$, but $G \in \TdMw{1}{0}$.
\end{lemma}

\begin{proof}
Consider the graph constructed in the proof of Lemma~\ref{lem:TDMWvsTDandMW} for $p = \ell$, cf.\ the left side of Figure~\ref{Fig:ExampleGraph}.
Since $G$ does not admit any non-trivial module and contains a clique of size $\ell$ it holds that $G \notin \Mtdl$, but $G \in \TdMw{1}{0}$ as shown in the proof of Lemma~\ref{lem:MTDandTDMWincomparable}.
\end{proof}

\begin{corollary}
For any $k, h \in\N$, there exists a graph $G$ such that $G \notin \TdkMwh$, but $G \in \MwMtd{0}{3}$.
\end{corollary}

\begin{proof}
This follows directly from Lemma~\ref{lem:MTDandTDMWincomparable}.
\end{proof}

We are left to compare the graph classes with the class $\TdkMtdl$.

\begin{corollary}
For any $h, \ell \in\N$, there exists a graph $G$ such that $G \notin \MwhMtdl$, but $G \in \TdMtd{1}{0}$.
\end{corollary}

\begin{proof}
This directly follows from Lemma~\ref{lem:TDMWsmallMTDlarge}.
\end{proof}

Clearly, for any graph $G \in \MwhMtdl$ it holds that $G \in \TdMtd{0}{\max\{h,\ell\}}$. Thus, it is not possible that for arbitrary $k, \ell \geq 0$ there exists graphs $G \notin \TdkMtdl$, but $G \in \MwMtd{h'}{\ell'}$ for constants $h'$ and  $\ell'$. However, it is possible to make such a claim with $\ell'$ being a constant. 

\begin{lemma}\label{lem:TdMtdToMw2l}
For any $k, \ell \in\N_{\geq 2}$, there exists a graph $G$ such that $G \notin \TdkMtdl$, but $G \in \MwMtd{2\ell}{0}$.
\end{lemma}

\begin{proof}
Let $H$ be a clique $K_{\ell}$ of size $\ell$ with a pendant vertex for each vertex in $K_\ell$ and let $G = \SubstH(K_{k+1}, \ldots, K_{k+1})$. Now, it holds that $G \notin \TdkMtdl$ but $\mw(G) = 2\ell$.
\end{proof}

Finally, we compare the graph class $\TdkMtdl$ with the class $\TdkMwh$.

\begin{corollary}
For any $k, h \in\N$, there exists a graph $G$ such that $G \notin \TdkMwh$, but $G \in \TdMtd{0}{3}$.
\end{corollary}

\begin{proof}
This directly follows from Lemma~\ref{lem:MTDandTDMWincomparable}.
\end{proof}

Clearly, for any graph $G \in \TdkMwh$ it holds that $G \in \TdkMtdl$ with $\ell = h$. Thus, it is again not possible that for arbitrary $k, \ell \geq 1$ there exists graphs $G \notin \TdkMtdl$, but $G \in \TdMw{k'}{h'}$ for constants $k'$ and  $h'$. However, it is possible to make such a claim with only $k'$ being a constant. 

\begin{corollary}
For any $k, \ell \in\N_{\geq 2}$, there exists a graph $G$ such that $G \notin \TdkMtdl$, but $G \in \TdMw{0}{2\ell}$.
\end{corollary}

\begin{proof}
This directly follows from Lemma~\ref{lem:TdMtdToMw2l}.
\end{proof}

\section{Conclusion}\label{section:conclusion}

Coping with the limited generality of graph structure for which we know fast algorithms, we have presented a clean and robust way of defining more general, heterogeneous structure via graph operations and algebraic expressions. We gave a generic framework that allows to design algorithms with competitive running times by designing routines for any subset of the considered operations. As applications we showed algorithms for \TC, \NCD, and \VWAPSP. 

In future work, we aim to extend this approach so that further operations and problems can be studied. This often touches upon concrete algorithmic questions for well-studied parameters. For example, it is an open problem to design a non-trivial algorithm for \prob{Maximum Matching} relative to modular tree-depth or, equivalently, for \prob{$b$-Matching} relative to tree-depth. Such an algorithm would yield a fast algorithm for \prob{Maximum Matching} for all graphs in $\TdkMwhMtdl$; as it stands, known algorithms~\cite{IwataOO18,KratschN18} yield time $\Oh((h^2 \log h)n + km )$ via Theorem~\ref{thm:runningtime}.

Another important spin-off question is to determine the complexity of finding for a given graph, a set of operations, and suitable parameters, e.g., $k,h,\ell\in\N$, an equivalent algebraic expression, i.e., a decomposition of the heterogeneous structure of the graph. Especially fast polynomial-time exact or approximate algorithms would be valuable. Nevertheless, because algebraic expressions over some set of graph operations are also graph constructions, we already learn what processes create graphs of beneficial heterogeneous structure.

\bibliography{main}

\end{document}